\documentclass[
	aps,
	pra,
	notitlepage,
	10pt,
	nofootinbib,
	floatfix,
]{revtex4-2}
\pdfoutput=1
\usepackage[utf8]{inputenc}
\usepackage[english]{babel}
\usepackage[T1]{fontenc}
\usepackage{mplmacros}

\usepackage{tikz}
\usetikzlibrary{cd}
\newcommand{\tao}[1]{{\color{black} #1}}
\newcommand{\update}[1]{{\color{black} #1}}

\DeclareMathOperator{\fix}{Fix}

\begin{document}
\title{Symmetries and Wigner representations of operational theories}

\author{Ties-A. Ohst \footnote{Author to whom any correspondence should be addressed}}
\email{ties-albrecht.ohst@uni-siegen.de}
\affiliation{Naturwissenschaftlich-Technische Fakult\"{a}t, Universit\"{a}t Siegen, Walter-Flex-Stra\ss e 3, 57068 Siegen, Germany}

\author{Martin Pl\'{a}vala}
%\email{martin.plavala@uni-siegen.de}
\affiliation{Naturwissenschaftlich-Technische Fakult\"{a}t, Universit\"{a}t Siegen, Walter-Flex-Stra\ss e 3, 57068 Siegen, Germany}
\affiliation{Institut f\"{u}r Theoretische Physik, Leibniz Universit\"{a}t Hannover, Hannover, Germany}

\begin{abstract}
We develop the theory of Wigner representations \update{for general probabilistic theories (GPTs),} a large class of operational theories that include both classical and quantum theory. The Wigner representations that we introduce are a natural way to describe the theory in terms of some fixed observables; these observables are often picked to be position and momentum or spin observables. This allows us to introduce symmetries which transform the outcomes of the observables used to construct the Wigner representation; we obtain several results for when these symmetries are well defined or when they uniquely specify the Wigner representation.
\end{abstract}

\maketitle

\section{Introduction}
Wigner functions as introduced by Wigner \cite{Wigner-WignerFunctions} are used as a description of quantum systems equivalent to the one of wave functions and density matrices. The main idea behind Wigner functions is that we want to describe the state of the particle in terms of position and momentum. In classical theory this is possible and we either describe a localized particle by its position and momentum or, in more general settings, we describe a particle using a probability density in phase space. Due to complementarity of position and momentum, quantum particles cannot in general be described by a probability density in phase space, but they can be described by a pseudo-probability density, that is, a quantum particle can be described by a density function in phase space with the caveat that said density function is not positive everywhere. Despite that, these density functions, usually called Wigner functions, have many useful properties that offer simple treatment of several quantum systems, see \cite{CurtrightFarlieZachos-wignerFunctions} for a review. Moreover negativity of Wigner functions is in fact what enables quantum advantage as negativity of Wigner function was identified as resource for several tasks \cite{RaussendorfBrowneDelfosseOkayBermejovega-wignerNegative, Rundle2021, Ferrie2011}. Recently, generalizations of Wigner functions were used in foundations of quantum theory to investigate physical systems describing non-classical and non-quantum harmonic oscillator and hydrogen atom \cite{PlavalaKleinmann-oscillator, PlavalaKleinmann-hydrogen}. Wigner functions were also developed for discrete observables such as spin and angular momentum \cite{Wootters-discreteWigner,GibbonsHoffmanWootters-discreteWigner,Gross-discreteWigner,DeBrotaStacey-discreteWigner,SchwonnekWerner-discreteWigner}; this has the benefit that all of the considered vector spaces are finite-dimensional \update{and the Wigner functions are defined over a finite phase space labelled by discrete variables \cite{Wootters-discreteWigner, Gross-discreteWigner,Ferrie2011, Muoz2017}} which significantly simplifies the underlying mathematics. \update{In concrete physical problems beyond the mathematical formalism, Wigner functions in discrete variable systems have been applied in spin systems \cite{Davis2021} and their connection to quantum contextuality has been discussed \cite{Delfosse2017, Raussendorf2020, Booth2022, haferkamp2021equivalence}. Additionally, recently, a different class of descriptions of quantum theory via approximations by the so-called $\Lambda$-polytopes was introduced in order to simulate quantum computation \cite{zurel2020hidden,okay2021extremal,zurel2023simulating,zurel2024hidden,ipek2023degenerate}. While the $\Lambda$-polytopes are reminiscent of Wigner functions for discrete systems, they are nevertheless different objects. Moreover, the $\Lambda$-polytopes do not simulate the quantum theory in it's entirety, but only a subset of quantum theory which is sufficient for universal computation. Hence, in order to theoretically grasp $\Lambda$-polytopes and similar objects, we need to widen our understanding of the fundamental aspects of Wigner functions of theories beyond quantum theory.}

Operational theories are a large class of theories that aim to include both classical and quantum theory. Most notable of these are the general probabilistic theories \cite{JanottaHinrichsen-review,Lami-thesis,Mueller-review,Plavala-review,Leppajarvi-thesis} which were previously used as a framework to derive quantum theory \cite{Hardy-derivationQT,MasanesMuller-derivationQT,ChiribellaDArianoPerinotti-derivationQT,Wilce-derivationQT,Buffenoir-derivationQT}, but also to investigate plethora of other topics, see \cite{Plavala-review} for an in-depth review. General probabilistic theories are based on the simple assumption that the state space, that is the set of all states of the system, should be convex since convex combinations represent probabilistic mixtures of preparation procedures. Based on this assumption one can derive rich structure of measurements, channels, tensor products, and investigate information-theoretic properties or requirements of various tasks.

Here we show how to construct Wigner representations of finite-dimensional general probabilistic theories that mimic the properties of Wigner functions used in quantum theory. \update{In comparison to the famous Stratonovic-Weyl criteria \cite{stratonovich1957distributions,Brif1998} that are an axiomatic foundation of what a Wigner function in quantum theory should be, our construction is based on minimal assumptions that especially do not include the presence of Hilbert spaces or covariance properties.} We investigate only the finite-dimensional case due to its simplicity, it is quite easy to construct counter-examples and build intuition for Wigner representations in the finite-dimensional case. We also investigate transformations and symmetries of the Wigner representations since these also play a crucial role in the theory of Wigner functions. Our main results on the structure of Wigner representations are that we identify necessary and sufficient conditions for the Wigner representation to be positive and when the Wigner representation can be chosen to be equivalent to the description using the state space. We then introduce symmetries of Wigner representations and we identify when they induce symmetries on the state space. Our final result is that we identify conditions under which the Wigner representation is unique; some of these conditions, such as complementarity, are natural generalizations of the properties of position and momentum observables.

The paper is organized as follows: in Section~\ref{sec:observables} we introduce general probabilistic theories, in Section~\ref{sec:Wigner} we define Wigner representations for operational theories, in Section~\ref{sec:symmetry} we introduce symmetries of Wigner representations, and in Section~\ref{sec:groups} we look at groups of symmetries. In Section~\ref{sec:unique} we formulate and prove the conditions for uniqueness of Wigner representations. In Section~\ref{sec:QT} we investigates examples of Wigner representations of quantum theory, more specifically of position and momentum of qudits and of mutually unbiased measurements, and in Section~\ref{sec:boxworld} we investigate examples of Wigner representations of the Boxworld theory, which is an example of non-classical and non-quantum operational theory.

\section{State spaces and observables} \label{sec:observables}
One can start introducing general probabilistic theories either by starting from the state space or from the effect algebra, we opt for the first option as it is conceptually simpler. A state space $K$ is a compact convex subset of a real, finite-dimensional vector space. Every point $x \in K$ represents a possible state of the system we are describing by $K$, where a state is an equivalence class of preparation procedures, see \cite{Plavala-review} for an elaborate explanation. \tao{Extreme points $k\in K$ of a state space are defined by the property that all convex decompositions $k = \lambda k_1 + (1-\lambda) k_2$ with $\lambda \in (0,1)$ imply that $k_1 = k_2$ and the set of extreme points is denoted as $\ext(K)$.} The effect algebra is the set of affine functions mapping $K$ to $[0, 1]$, we require the functions to be affine in order to preserve the convex structure of $K$. That is, $E(K) = \{ f: K \to [0,1]: f(\lambda x + (1-\lambda)y) = \lambda f(x) + (1-\lambda) f(y), \forall x,y \in K, \forall \lambda \in [0, 1]\}$. An important example of an effect is the unit effect $1_K$ defined as $1_K(x) = 1$ for all $x \in K$, it is straightforward to check that $1_K \in E(K)$.

Let $K_1$ and $K_2$ be state spaces, a channel is an affine map $\Phi: K_1 \to K_2$, that is, for every $x, y \in K_1$ and $\lambda \in [0,1]$ we have $\Phi(x) \in K_2$ and $\Phi(\lambda x + (1-\lambda) y ) = \lambda \Phi(x) + (1-\lambda) \Phi(y)$. Note that in quantum theory, we usually request that channels are completely positive, i.e., that they map entangled states to entangled state. But since we are working only with single systems in hypothetical theories, we do not need to consider complete positivity of channels. An important example of a channel is the identity channel $\id_K: K \to K$ defined as $\id_K(x) = x$ for all $x \in K$.

In an operational theory modeled using a state space $K$, we usually describe the process of measuring an observable using an affine map $\mu: K \to \Prob(\Omega)$, where $\Omega$ is some finite set of possible measurement outcomes and $\Prob(\Omega)$ is the set of probability distributions on $\Omega$. Each element of $\Prob(\Omega)$ can be represented as vector of probabilities of the form $(p_a)_{a \in \Omega}$ where $p_a$ is the probability of outcome $a$ occurring, or, equivalently, measure of the set $\{a\}$. A special case of such a vector is the constant probability distribution $\tau = (\frac{1}{\abs{\Omega}})_{a \in \Omega}$ which assigns the same probability to all possible outcomes, here $\abs{\Omega}$ denotes the cardinality of $\Omega$. We thus have that $\mu: K \to \Prob(\Omega)$ maps a state $x \in K$ to a vector of probabilities $\mu(x) = (p_a(x))_{a \in \Omega}$, $p_a(x) \in \Rp$, and $\sum_{a \in \Omega} p_a(x) = 1$. We usually assume that $\mu$ preserves the convex structure of $K$, meaning that for $x,y \in K$ and $\lambda \in [0,1]$ we have $\mu(\lambda x + (1-\lambda) y) = \lambda \mu(x) + (1-\lambda) \mu(y)$. It follows that we must have $p_a(\lambda x + (1-\lambda) y) = \lambda p_a(x) + (1-\lambda) p_a(y)$ for all $a \in \Omega$. One can now deduce that $p_a: K \to \Rp$ is an affine function and that in fact $p_a \in E(K)$. This reconstructs the well known result that one can equivalently describe $\mu: K \to \Prob(\Omega)$ as a tuple of effects $(p_a)_{a \in \Omega} \subset  E(K)$. We thus have the following definition:
\begin{definition}
An observable $\oA$ on $K$ is a pair $(\Omega_{\oA}, \mu_{\oA})$, where $\Omega_{\oA} \subset \RR$ is a nonempty subset with finite cardinality and $\mu_{\oA}$ is an affine map $\mu_{\oA}: K \to \Prob(\Omega_{\oA})$.
\end{definition}
$\Omega_{\oA}$ is usually called the sample space of the observable $\oA$ and we can always assume that $\Omega_{\oA}$ is a measurable space equipped with the $\sigma$-algebra of all subsets of $\Omega_{\oA}$. One can in principle drop the assumption that $\Omega_{\oA}$ must have finite cardinality and simply assume that $\Omega_{\oA}$ is an arbitrary measurable space. We will for simplicity keep the assumption of finite cardinality of $\Omega_{\oA}$.

Let $\oA = (\Omega_{\oA}, \mu_{\oA})$ and $\oB = (\Omega_{\oB}, \mu_{\oB})$ be two observables on $K$ and consider a map $m: K \to \Prob(\Omega_{\oA} \times \Omega_{\oB})$, where $\Omega_{\oA} \times \Omega_{\oB}$ is the Cartesian product of $\Omega_{\oA}$ and $\Omega_{\oB}$. Let $x \in K$, then $m(x)$ can be expressed as a vector of probabilities indexed by $a \in \Omega_{\oA}$ and $b \in \Omega_{\oB}$, i.e., we have $m(x) = ( p_{ab}(x) )_{a \in \Omega_{\oA}, b \in \Omega_{\oB}}$. We can then define the marginals of $m$ by integrating over $\Omega_{\oA}$ or $\Omega_{\oB}$, that is, we can define $\int_{\Omega_{\oB}} m(x) \dd b = ( \sum_{b \in \Omega_{\oB}} p_{ab}(x) )_{a \in \Omega_{\oA}}$ and $\int_{\Omega_{\oA}} m(x) \dd a = ( \sum_{a \in \Omega_{\oA}} p_{ab}(x) )_{b \in \Omega_{\oB}}$. We then say that $\oA$ and $\oB$ are compatible if $m: K \to \Prob(\Omega_1 \times \Omega_2)$ can be chosen such that for all $x \in K$ we have
\begin{align}
\int_{\Omega_{\oB}} m(x) \dd b &= \mu_{\oA}(x), \label{eq:observables-marginal1} \\
\int_{\Omega_{\oA}} m(x) \dd a &= \mu_{\oB}(x). \label{eq:observables-marginal2}
\end{align}
Compatible observables are also called jointly measurable, and the operational interpretation is that although given a single copy of the input state $x \in K$ we can not measure both $\oA$ and $\oB$, but we can use $m$ from which we can obtain the statistics of $\oA$ and $\oB$ as marginals of $m$. Compatibility of observables was previous heavily investigated \cite{BuschHeinosaariSchultzStevens-compatibility,HeinosaariMiyaderaZiman-compatibility,FilippovHeinosaariLeppajarvi-compatibility,Jencova-compatibility,Kuramochi-simplex,BluhmJencovaNechita-compatibility} and is known to be connected to broadcasting \cite{BarnumBarrettleiferWilce-noBroadcasting,Heinosaari-compatBroadcast} and non-local phenomena such as steering \cite{UolaCostaNguyenGuhne-steering,HobanSainz-channels,CavalcantiSelbySikoraGalleySainz-witworld} and violations of Bell inequalities \cite{BrunnerCavalcantiPironioScaraniWehner-BellIneq,RossetBancalGisin-BellIneq}.

One can also investigate whether two observables provide enough information to uniquely specify which state was measured. Hence the following definition:
\begin{definition}
Let $K$ be a state space and let $\oA$, $\oB$ be observables on $K$. We say that $\oA$ and $\oB$ are jointly info-complete if for any $x,y \in K$ we have that $\mu_{\oA}(x) = \mu_{\oA}(y)$ and $\mu_{\oB}(x) = \mu_{\oB}(y)$ implies $x=y$.
\end{definition}

\section{Wigner representations of operational theories} \label{sec:Wigner}
Let $K$ be a state space and let $\oA$ and $\oB$ be observables describing physical properties of states. One can then ask how to describe states from $K$ in terms of values of the observables $\oA$ and $\oB$? For example, in classical mechanics, we describe almost exclusively all systems in terms of the position and momentum observables, while in thermodynamics we can describe the same system in terms of various thermodynamic variables. In a similar fashion, we aim to present a way to describe states of $K$ in terms of the observables $\oA$ and $\oB$. To do so, we first have to observe that in general, not every state $x \in K$ has a sharp value of $\oA$ or $\oB$, and so it follows that we can not simply describe a state $x \in K$ using a pair of points $(a,b) \in \Omega_{\oA} \times \Omega_{\oB}$. One can then propose to use the pair $(\mu_{\oA}(x), \mu_{\oB}(x)) \in \Prob(\Omega_{\oA}) \times \Prob(\Omega_{\oB})$ to describe $x \in K$, but this would imply that we have access to both $\mu_{\oA}(x)$ and $\mu_{\oB}(x)$ at the same time. This would suggest that we can measure both $\oA$ and $\oB$ jointly, which is not true in general, think for example of the position and momentum observables in quantum theory. Hence we opt for the following approach: let $\SProb(\Omega_{\oA} \times \Omega_{\oB})$ denote the set of signed probability measures on $\Omega_{\oA} \times \Omega_{\oB}$, that is, $\nu \in \SProb(\Omega_{\oA} \times \Omega_{\oB})$ is a vector of numbers $\nu = (q_{ab})_{a \in \Omega_{\oA}, b \in \Omega_{\oB}}$ such that $q_{ab} \in \RR$ for all $a \in \Omega_{\oA}$, $b \in \Omega_{\oB}$ and $\sum_{a \in \Omega_{\oA}, b \in \Omega_{\oB}} q_{ab} = 1$. For every $\nu \in \SProb(\Omega_{\oA} \times \Omega_{\oB})$ one can define the marginals by integrating either over $\Omega_{\oA}$ and $\Omega_{\oB}$ as $\int_{\Omega_{\oB}} \nu \dd b = ( \sum_{b \in \Omega_{\oB}} q_{ab} )_{a \in \Omega_{\oA}}$, $\int_{\Omega_{\oA}} \nu \dd a = ( \sum_{a \in \Omega_{\oA}} q_{ab} )_{b \in \Omega_{\oB}}$, and one can easily check that we have $\int_{\Omega_{\oA}} \nu \dd a \in \SProb(\Omega_{\oA})$ and $\int_{\Omega_{\oB}} \nu \dd b \in \SProb(\Omega_{\oB})$. We now define the Wigner representation of $K$ corresponding to observables $\oA$ and $\oB$ as follows:
\begin{definition}
Let $K$ be a state space and let $\oA$ and $\oB$ be observables. Then a Wigner representation of $K$ corresponding to $\oA$ and $\oB$ is an affine map $W: K \to \SProb(\Omega_{\oA} \times \Omega_{\oB})$ such that for every $x \in K$ we have
\begin{align}
\int_{\Omega_{\oB}} W(x) \dd b &= \mu_{\oA}(x), \label{eq:Wigner-marginalA} \\
\int_{\Omega_{\oA}} W(x) \dd a &= \mu_{\oB}(x). \label{eq:Wigner-marginalB}
\end{align}
\end{definition}
\tao{Wigner representations with respect to two observables can conveniently expressed as matrices
\begin{equation}
W(x) = 
\begin{pmatrix}
W(x)_{a_1,b_1} & \dots &  W(x)_{a_1 \abs{\Omega_{\oB}}} \\
\vdots & \ddots & \vdots \\
W(x)_{\abs{\Omega_{\oA}} b_1} & \dots & W(x)_{\abs{\Omega_{\oA}} \abs{\Omega_{\oB}}}
\end{pmatrix}
\end{equation}
where $\mu_{\oA}(x)$ corresponds to the vector of row sums of $W$ and $\mu_{\oB}(x)$ to the vector of column sums of $W$.
}

The Wigner representation is heavily inspired by Wigner functions \cite{Wigner-WignerFunctions,Groenewold-QM,Moyal-WignerFunctions,MuckenheimLudwigEtAl-WignerFunctions,HilleryConnellScullyWigner-WignerFunctions}. Let $\psi$ be a wave-function, then the corresponding Wigner function is a function on the phase space $\RR^2 = \{ (q,p): q,p \in \RR \}$ defined as
\begin{equation} \label{eq:Wigner-PQwigner}
W(q,p) = \frac{1}{h} \int_{\RR} \e^{-\frac{ipy}{\hbar}} \psi(q+ \frac{y}{2}) \psi^*(q-\frac{y}{2}) \dd y.
\end{equation}
It is well-known that the marginals of the Wigner functions are distributions of position and momentum, see \cite{Case-WignerFunctions}. Wigner functions were investigated before \cite{Baker-WignerFunctions,Cohen-WignerFunctions,Rosen-WignerFunctions}, they were used in many different contexts \cite{WeinbubFerry-WignerFunctions}, their analogs in finite-dimensional theories were investigated as well \cite{Wootters-discreteWigner,GibbonsHoffmanWootters-discreteWigner,Gross-discreteWigner,DeBrotaStacey-discreteWigner,SchwonnekWerner-discreteWigner} and they were also used in the context of operational theories \cite{Spekkens-epistricted,PlavalaKleinmann-oscillator,PlavalaKleinmann-hydrogen}. Our aim is to investigate  the properties of symmetries with respect to the Wigner representation, but before we do so, we will introduce several general properties of the Wigner representations.

First of all, note that one can easily extend the definition to $n$ observables $\oA_1, \ldots, \oA_n$ by defining the Wigner representation as an affine map $W: K \to \SProb(\Omega_{\oA_1} \times \ldots \times \Omega_{\oA_n})$ such that $\int_{\prod_{j \neq i} \Omega_{\oA_j}} W(x) = \mu_{\oA_i}(x)$ for all $i \in \{1, \ldots, n\}$ and where $\prod_{j \neq i} \Omega_{\oA_j} = \Omega_{\oA_1} \times \ldots \times \Omega_{\oA_{i-1}} \times \Omega_{\oA_{i+1}} \times \ldots \Omega_{\oA_n}$. Also note that in general $W$ is not unique.
\begin{proposition} \label{prop:Wigner-unique}
Let $K$ be a state space and let $\oA_1, \ldots, \oA_n$ be observables. Then the Wigner representation of $K$ corresponding to $\oA_1, \ldots, \oA_n$ is unique only if there is at most one observable $\oA_j$ such that $\abs{\Omega_j} > 1$, where $\abs{\Omega_j}$ is the cardinality of $\Omega_j$.
\end{proposition}
\begin{proof}
Without loss of generality assume that we have $\abs{\Omega_{\oA_1}} > 1$ and $\abs{\Omega_{\oA_2}} > 1$. Let $W: K \to \SProb(\Omega_{\oA_1} \times \ldots \times \Omega_{\oA_n})$ be the respective Wigner representation and let $x \in K$, then $W(x)$ is a vector of numbers $W(x) = ( q_{a_1, \ldots, a_n}(x) )_{a_1 \in \Omega_{\oA_1}, \ldots, a_n \in \Omega_{\oA_n}}$. Let $\alpha_{11}, \alpha_{21} \in \Omega_{\oA_1}$, $\alpha_{11} \neq \alpha_{21}$, $\alpha_{12}, \alpha_{22} \in \Omega_{\oA_2}$, $\alpha_{12} \neq \alpha_{22}$ and let $r_{a_1, \ldots, a_n} \in \RR$, $a_1 \in \Omega_{\oA_1}, \ldots, a_n \in \Omega_{\oA_n}$, be numbers defined as $r_{\alpha_{11},\alpha_{12},a_3, \ldots, a_n} = r_{\alpha_{21},\alpha_{22},a_3, \ldots, a_n} = 1$ and $r_{\alpha_{11},\alpha_{22},a_3, \ldots, a_n} = r_{\alpha_{21},\alpha_{12},a_3, \ldots, a_n} = -1$ for all $a_3 \in \Omega_{\oA_3}, \ldots, a_n \in \Omega_{\oA_n}$ and for $\alpha_{11} \neq a_1 \neq \alpha_{21}$, $\alpha_{12} \neq a_2 \neq \alpha_{22}$ as $r_{a_1, \ldots, a_n} = 0$. The important point here is to observe that
\begin{equation}
\sum_{a_2 \in \Omega_{\oA_2}, \ldots, a_n \in \Omega_{\oA_n}} r_{a_1, \ldots, a_n}
= \sum_{a_1 \in \Omega_{\oA_1}, a_3 \in \Omega_{\oA_3}, \ldots, a_n \in \Omega_{\oA_n}} r_{a_1, \ldots, a_n}
= \ldots
= \sum_{a_1 \in \Omega_{\oA_1}, \ldots, a_{n-1} \in \Omega_{\oA_{n-1}}} r_{a_1, \ldots, a_n}
= 0,
\end{equation}
i.e., that all marginals are $0$. Now, define $W': K \to \SProb(\Omega_{\oA_1} \times \ldots \times \Omega_{\oA_n})$ as follows: let $x \in K$, then $W'(x) = ( q_{a_1, \ldots, a_n}(x) + r_{a_1, \ldots, a_n} )_{a_1 \in \Omega_{\oA_1}, \ldots, a_n \in \Omega_{\oA_n}}$. It follows that $W'$ has the same marginals as $W$ and so $W'$ is a Wigner representation of $K$ with respect to $\oA$ and $\oB$ different from $W$.
\end{proof}
From now on we will for simplicity restrict only to the case of two observables $\oA$ and $\oB$. We will now show that the non-uniqueness of Wigner representations is non-trivial, as some Wigner representations corresponding to fixed observables $\oA$ and $\oB$ can carry more information about the underlying state space than other Wigner representations corresponding to the same $\oA$ and $\oB$.
\begin{example}[Wigner representations of a cube]
Let $C$ be a cube in $\RR^3$ with vertices $s_{ijk}$, $i,j,k \in \{0, 1\}$; one can identify $s_{ijk}$ with a vector $(i,j,k)^\intercal \in \RR^3$. Let $f_x, f_y, f_z \in E(C)$ be given as $f_x(s_{ijk}) = i$, $f_y(s_{ijk}) = j$, $f_z(s_{ijk}) = k$ for all $i,j,k \in \{0, 1\}$, and let $\oA$ and $\oB$ be observables given as follows: $\Omega_{\oA} = \Omega_{\oB} = \{0, 1\}$ and for $s \in C$ we have $\mu_{\oA}(s) = ( f_x(s), 1-f_x(s) )$, $\mu_{\oB}(s) = ( f_y(s), 1-f_y(s) )$. One can easily see that any Wigner representation of $C$ with respect to $\oA$ and $\oB$ is of the form
\begin{equation}
W(s) =
\begin{pmatrix}
q(s) & f_x(s) - q(s) \\
f_y(s) - q(s) & 1 + q(s) - (f_x(s) + f_y(s))
\end{pmatrix}
\end{equation}
where $s \in C$ and $q: C \to \RR$ is an affine function. Note that we are using a matrix to represent $W$ only to make our notation compact and because then one reconstructs $\oA$ and $\oB$ by summing the rows and columns of $W$, respectively.

Choosing $q(s) = 0$ for all $s \in C$, we get the Wigner representation $W_0$ and we have $W_0(s_{ij0}) = W_0(s_{ij1})$ for all $i,j \in \{0,1\}$. Choosing $q = f_z$ we get the Wigner representation $W_z$ and one can see that for any $s, s' \in C$ we have that $W_z(s) = W_z(s')$ implies $f_x(s) = f_x(s')$, $f_y(s) = f_y(s')$, and $f_z(s) = f_z(s')$, which implies $s = s'$. It follows that $W_z$ carries more information about $C$ than $W_0$, despite the fact that they are both Wigner representations with respect to the same two observables $\oA$ and $\oB$.
\end{example}

We will now introduce two additional useful properties of Wigner representations.
\begin{definition}
Let $K$ be a state space, let $\oA$ and $\oB$ be observables and let $W$ be the corresponding Wigner representation. We say that $W$ is positive if $W(x) \in \Prob(\Omega_{\oA} \times \Omega_{\oB})$ for all $x \in K$ and we say that $W$ is faithful if for any $x,y \in K$ we have that $W(x) = W(y)$ implies $x = y$.
\end{definition}
The following result is immediate, a similar result was already observed by Wigner in \cite{Wigner-WignerFunctions} but also recently in \cite{GhaiSharmaGhos-wignerIncompatibility}. The intuition is that if $W$ is positive, then we can in fact treat $W$ as an observable with sample space $\Omega_{\oA} \times \Omega_{\oB}$ and the corresponding map defined for $x \in K$ as $m(x) = W(x)$. It then follows from \eqref{eq:Wigner-marginalA} and \eqref{eq:Wigner-marginalB} that $\mu_{\oA}$ and $\mu_{\oB}$ are marginals of $m$, implying that the measurements $\mu_{\oA}$ and $\mu_{\oB}$ are compatible.
\begin{theorem} \label{thm:Wigner-positive}
Let $K$ be a state space and let $\oA$ and $\oB$ be observables. Then there exists a positive Wigner representation of $K$ corresponding to $\oA$ and $\oB$ if and only if the measurements $\mu_{\oA}$ and $\mu_{\oB}$ are compatible.
\end{theorem}
\begin{proof}
The proof is straightforward: let $W$ be a positive Wigner representation of $K$ with respect to $\oA$ and $\oB$ and define a measurement $m: K \to \Prob(\Omega_{\oA}\times\Omega_{\oB})$ as $m(x) = W(x)$ for all $x \in K$. Then \eqref{eq:observables-marginal1} and \eqref{eq:observables-marginal2} follow from \eqref{eq:Wigner-marginalA} and \eqref{eq:Wigner-marginalB}. If $\mu_{\oA}$ and $\mu_{\oB}$ are compatible measurements, then let $m: K \to \Prob(\Omega_{\oA} \times \Omega_{\oB})$ be the respective joint measurement and define $W: K \to \SProb(\Omega_{\oA} \times \Omega_{\oB})$ as $W(x) = m(x)$. It follows that $W$ is a Wigner representation of $K$ with respect to $\oA$ and $\oB$, since \eqref{eq:Wigner-marginalA} and \eqref{eq:Wigner-marginalB} follow from \eqref{eq:observables-marginal1} and \eqref{eq:observables-marginal2}. Moreover it follows that $W$ is positive.
\end{proof}
\update{One can use the result of Theorem~\ref{thm:Wigner-positive} to connect positive Wigner representations to the existence of non-contextual hidden variable models by using known results that connect contextuality and incompatibility \cite{tavakoli2020measurement,plavala2022incompatibility}. In this way Theorem~\ref{thm:Wigner-positive} provides a direct generalization of the previous results connecting negativity of Wigner representations and contextuality that was already observed in quantum theory \cite{Delfosse2017,Raussendorf2020,Booth2022,haferkamp2021equivalence}.}

Faithful Wigner representation are of separate interest. We have already seen that in some cases for fixed observables $\oA$ and $\oB$ some Wigner representations are faithful while other are not. Here we show that this is the case unless $\oA$ and $\oB$ are together informationally-complete.
\begin{proposition} \label{prop:Wigner-allFaithful}
Let $K$ be a state space and let $\oA$ and $\oB$ be observables. Then all Wigner representation of $K$ corresponding to $\oA$ and $\oB$ are faithful if and only if $\oA$ and $\oB$ are jointly info-complete.
\end{proposition}
\begin{proof}
Let $x,y \in K$, then $W(x) = W(y)$ implies $\mu_{\oA}(x) = \mu_{\oA}(y)$ and $\mu_{\oB}(x) = \mu_{\oB}(y)$. It follows that if $\mu_{\oA}(x) = \mu_{\oA}(y)$ and $\mu_{\oB}(x) = \mu_{\oB}(y)$ implies $x=y$, then also $W(x) = W(y)$ implies $x=y$. Now assume that there are $x, y \in K$ such that $x \neq y$, but $\mu_{\oA}(x) = \mu_{\oA}(y)$ and $\mu_{\oB}(x) = \mu_{\oB}(y)$; we will now construct a Wigner representation $W$ such that $W(x) = W(y)$. Let $\mu_{\oA}(x) = ( f_{a, \oA}(x) )_{a \in \Omega_{\oA}}$ and $\mu_{\oB}(x) = ( f_{b, \oB}(x) )_{b \in \Omega_{\oB}}$ where $f_{a, \oA}, f_{b, \oB} \in E(K)$ for all $a \in \Omega_{\oA}$ and $b \in \Omega_{\oB}$ and $\sum_{a \in \Omega_{\oA}} f_{a, \oA} = \sum_{b \in \Omega_{\oB}} f_{b, \oB} = 1_K$. We then define $W: K \to \SProb(\Omega_{\oA}\times\Omega_{\oB})$ for $z \in K$ as $W(z) = (q_{ab}(z))_{a \in \Omega_{\oA}, b \in \Omega_{\oB}}$ where $q_{ab}$ are defined as follows: pick fixed elements $\alpha \in \Omega_{\oA}$ and $\beta \in \Omega_{\oB}$, then
\begin{equation}
q_{ab} =
\begin{cases}
f_{a, \oA} & a \neq \alpha, b = \beta \\
f_{b, \oB} & a = \alpha, b \neq \beta \\
1_K - \sum_{a \in \Omega_{\oA} \setminus \{\alpha\}} f_{a, \oA} - \sum_{b \in \Omega_{\oB} \setminus \{\beta\}} f_{b, \oB} & a = \alpha, b = \beta \\
0 & a \neq \alpha, b \neq \beta
\end{cases}
\end{equation}
It is straightforward to verify that $W$ is a Wigner representation of $K$ with respect to $\oA$ and $\oB$. Now let $x,y \in K$, $x \neq y$, be the states such that $\mu_{\oA}(x) = \mu_{\oA}(y)$ and $\mu_{\oB}(x) = \mu_{\oB}(y)$. It follows that then we must have $f_{a, \oA}(x) = f_{a, \oA}(y)$ and $f_{b, \oB}(x) = f_{b, \oB}(y)$ for all $a \in \Omega_{\oA}$ and $b \in \Omega_{\oB}$. But then $q_{ab}(x) = q_{ab}(y)$ and so $W(x) = W(y)$ and so $W$ is not faithful.
\end{proof}
When can we choose $W$ to be faithful? In some cases this ought to be impossible, but in other cases this has to be only matter of choosing the right Wigner representation. The following result investigates when we can choose $W$ to be faithful.
\begin{theorem}
Let $K$ be a state space, let $\oA$ and $\oB$ be observables. Let $x \in K$ and let the underlying measurements $\mu_{\oA}$ and $\mu_{\oB}$ be given as $\mu_{\oA}(x) = ( f_{a, \oA}(x) )_{a \in \Omega_{\oA}}$, $\mu_{\oB}(x) = ( f_{b, \oB}(x) )_{b \in \Omega_{\oB}}$, where $f_{a, \oA}, f_{b, \oB} \in E(K)$ for all $a \in \Omega_{\oA}$ and $b \in \Omega_{\oB}$ and $\sum_{a \in \Omega_{\oA}} f_{a, \oA} = \sum_{b \in \Omega_{\oB}} f_{b, \oB} = 1_K$. Then the Wigner representation $W$ can be chosen to be faithful if and only if
\begin{equation} \label{eq:Wigner-faitfulIneq}
(\abs{\Omega_{\oA}} - 1)(\abs{\Omega_{\oB}} - 1) \geq \dim(K) + 1 - \dim(\linspan(\{f_{a, \oA}\}_{a \in \Omega_{\oA}} \cup \{f_{b, \oB}\}_{b \in \Omega_{\oB}}))
\end{equation}
where $\abs{\Omega_{\oA}}$, $\abs{\Omega_{\oB}}$ is the cardinality of $\Omega_{\oA}$, $\Omega_{\oB}$ respectively.
\end{theorem}
\begin{proof}
Let $W$ be a Wigner representation of $K$ with respect to $\oA$ and $\oB$, then for $x \in K$ we have $W(x) = ( q_{ab}(x) )_{a \in \Omega_{\oA}, b \in \Omega_{\oB}}$ where $q_{ab}: K \to \RR$ are affine functions. It is clear that $W$ is faithful if and only if $q_{ab}$ span the vector space of affine functions on $K$, i.e., if
\begin{equation} \label{eq:Wigner-faithfulIneq-dimqab}
\dim(\linspan(\{q_{ab}\}_{a \in \Omega_{\oA}, b \in \Omega_{\oB}})) = \dim(K) + 1,
\end{equation}
see \cite{SingerStulpe-phaseSpace}. It follows from \eqref{eq:Wigner-marginalA} and \eqref{eq:Wigner-marginalB} that
\begin{align}
\sum_{b \in \Omega_{\oB}} q_{ab} &= f_{a, \oA}, \label{eq:Wigner-faithfulIneq-sumfa} \\
\sum_{a \in \Omega_{\oA}} q_{ab} &= f_{b, \oB}, \label{eq:Wigner-faithfulIneq-sumfb}
\end{align}
so $f_{a, \oA}, f_{b, \oB} \in \linspan(\{q_{ab}\}_{a \in \Omega_{\oA}, b \in \Omega_{\oB}})$ and we have $\linspan((\{f_{i, \oA}\}_{i=1}^n \cup \{f_{j, \oB}\}_{j=1}^m)) \subset \linspan(\{q_{ab}\}_{a \in \Omega_{\oA}, b \in \Omega_{\oB}})$. We will show that we can freely choose exactly $(\abs{\Omega_{\oA}} - 1)(\abs{\Omega_{\oB}} - 1)$ of $q_{ab}$ to maximize $\dim(\linspan(\{q_{ab}\}_{a \in \Omega_{\oA}, b \in \Omega_{\oB}}))$. Assume that is true, then we have $(\abs{\Omega_{\oA}} - 1)(\abs{\Omega_{\oB}} - 1) + \dim(\linspan(\{f_{a, \oA}\}_{a \in \Omega_{\oA}} \cup \{f_{b, \oB}\}_{b \in \Omega_{\oB}})) \geq \dim(\linspan(\{q_{ab}\}_{a \in \Omega_{\oA}, b \in \Omega_{\oB}}))$ and it follows that \eqref{eq:Wigner-faitfulIneq} is necessary and sufficient to satisfy \eqref{eq:Wigner-faithfulIneq-dimqab}.

To show that we can freely choose $(\abs{\Omega_{\oA}} - 1)(\abs{\Omega_{\oB}} - 1)$ of the functions $q_{ab}$, consider the following: pick fixed elements $\alpha \in \Omega_{\oA}$ and $\beta \in \Omega_{\oB}$, let $q'_{ab}: K \to \RR$, $a \in \Omega_{\oA}\setminus\{\alpha\}$, $b \in \Omega_{\oB}\setminus\{\beta\}$ be any affine functions and let
\begin{equation}
q_{ab} =
\begin{cases}
f_{a, \oA} - \sum_{b \in \Omega_{\oB} \setminus \{\beta\}} q'_{ab}  & a \neq \alpha, b = \beta \\
f_{b, \oB} - \sum_{a \in \Omega_{\oA} \setminus \{\alpha\}} q'_{ab} & a = \alpha, b \neq \beta \\
1_K - \sum_{a \in \Omega_{\oA} \setminus \{\alpha\}} f_{a, \oA} - \sum_{b \in \Omega_{\oB} \setminus \{\beta\}} f_{b, \oB} + \sum_{a \in \Omega_{\oA} \setminus \{\alpha\}, b \in \Omega_{\oB} \setminus \{\beta\}} q'_{ab} & a = \alpha, b = \beta \\
q'_{ab} & a \neq \alpha, b \neq \beta
\end{cases}
\end{equation}
It is straightforward to check that $q_{ab}$ form a Wigner representation of $K$ with respect to $\oA$ and $\oB$. Moreover it follows that we can freely choose $(\abs{\Omega_{\oA}} - 1)(\abs{\Omega_{\oB}} - 1)$ functions in the Wigner representation, these are exactly the $q'_{ab}$. It also follows that one function in every row and column is fixed by \eqref{eq:Wigner-faithfulIneq-sumfa} and \eqref{eq:Wigner-faithfulIneq-sumfb}, so one can freely choose at most $(\abs{\Omega_{\oA}} - 1)(\abs{\Omega_{\oB}} - 1)$ functions.
\end{proof}
Note that one can make $\abs{\Omega_{\oA}}$ and $\abs{\Omega_{\oB}}$ arbitrary large by simply adding outcomes that never occur. One can in this way always satisfy the inequality \eqref{eq:Wigner-faitfulIneq} and so for any observables $\oA$ and $\oB$ one can construct some faithful Wigner representation $W$, with the caveat that one might need to artificially enlarge $\Omega_{\oA}$ or $\Omega_{\oB}$. Moreover one can show that all faithful Wigner representations are isomorphic to each other.
\begin{proposition}
Let $W_1$ and $W_2$ be faithful Wigner representations of $K$ with respect to pairs of observables $\oA_1, \oB_1$ and $\oA_2, \oB_2$. Then $W_1$ and $W_2$ are isomorphic, i.e., there is a bijective map $\Lambda: \SProb(\Omega_{\oA_1} \times \Omega_{\oB_1}) \to \SProb(\Omega_{\oA_2} \times \Omega_{\oB_2})$ such that $\Lambda(W_1(x)) = W_2(x)$ for all $x \in K$.
\end{proposition}
\begin{proof}
If $W_1$ and $W_2$ are faithful, then the maps $W_1^{-1}: W_1(K) \to K$ and $W_2^{-1}: W_2(K) \to K$ are well-defined and we have $W_1^{-1} \circ W_1 = W_2^{-1} \circ W_2 = \id_K$. It then suffices to take $\Lambda = W_2 \circ W_1^{-1}$.
\end{proof}

\section{Symmetries} \label{sec:symmetry}
We can now discuss the connection between symmetries and Wigner representations of state spaces. In general, let $\Lambda: \SProb(\Omega_{\oA} \times \Omega_{\oB}) \to \SProb(\Omega_{\oA} \times \Omega_{\oB})$ be an affine map between the signed measures, then we say that $\Lambda$ is a symmetry of the Wigner representation $W$ of $K$ with respect to observables $\oA$ and $\oB$ if we have $\Lambda: W(K) \to W(K)$ i.e., if for every $x \in K$ we have $\Lambda(W(x)) \in W(K)$. It is obvious that if $\Lambda_1: \SProb(\Omega_{\oA} \times \Omega_{\oB}) \to \SProb(\Omega_{\oA} \times \Omega_{\oB})$ and $\Lambda_2: \SProb(\Omega_{\oA} \times \Omega_{\oB}) \to \SProb(\Omega_{\oA} \times \Omega_{\oB})$ are symmetries of $W$, then also their composition $\Lambda_2 \circ \Lambda_1$ is a symmetry of $W$. We will identify two distinct classes of symmetries of a Wigner representation, these are lifted symmetries and transported symmetries. We will start with lifted symmetries, that are native to the phase space, but to define them properly, we have to introduce the concept of lifted maps.

Let $X = \{1, \ldots, n \}$ be a finite set and let $\varphi: X \to X$ be a map. We will define a lifted map $\varphi^{\uparrow}: \SProb(X) \to \SProb(X)$ as follows: let $\nu \in \SProb(X)$ be given as $\nu = (q_1, \ldots, q_n)$, then we define $\varphi^\uparrow \nu \in \SProb(X)$ as
\begin{equation}
\varphi^\uparrow \nu = \left( \sum_{\varphi(j)=1} q_j, \sum_{\varphi(j)=2} q_j, \ldots, \sum_{\varphi(j)=n} q_j \right) = \left( \sum_{\varphi(j)=i} q_j \right)_{i \in X}
\end{equation}
To see that $\varphi^\uparrow \nu$ is normalized, simply observe that for every $j \in X$ there is exactly one $i \in X$ such that $\varphi(j) = i$. So we have $\sum_{i \in X} \sum_{\varphi(j)=i} q_j = \sum_{j \in X} q_j = 1$ and thus $\varphi^\uparrow \nu$ is properly normalized. Moreover we can prove that $\varphi^\uparrow: \SProb(X) \to \SProb(X)$ is an affine map, the proof is straightforward: let $q,r \in \SProb(X)$ and $\alpha \in \RR$, then we have
\begin{equation}
\begin{split}
\varphi^\uparrow (\alpha q + (1-\alpha) r) &= \left( \sum_{\varphi(j)=i} \alpha q_j + (1-\alpha) r_j \right)_{i \in X} \\
&= \alpha \left( \sum_{\varphi(j)=i} q_j \right)_{i \in X} + (1-\alpha) \left( \sum_{\varphi(j)=i} r_j \right)_{i \in X} \\
&= \alpha \varphi^\uparrow q + (1-\alpha) \varphi^\uparrow r
\end{split}
\end{equation}
We also have that if $\nu \in \Prob(X)$ then $\varphi^\uparrow \nu \in \Prob(X)$. Let $\nu \in \Prob(X)$, then $\nu = ( p_i )_{i \in X}$ where $p_i \geq 0$. But then also $\sum_{\varphi(j)=i} p_j \geq 0$ for all $i \in X$ and so $\varphi^\uparrow \nu \in \Prob(X)$. One can also observe that not every affine map $\Lambda: \SProb(X) \to \SProb(X)$ is lifted, for example by constructing a map such that $\Lambda \nu \notin \Prob(X)$ for some $\nu \in \Prob(X)$.

\begin{definition}
Let $K$ be a state space, let $\oA$ and $\oB$ be observables and let $W$ be a Wigner representation of $K$ with respect to $\oA$ and $\oB$. Let $\varphi: \Omega_{\oA}\times\Omega_{\oB} \to \Omega_{\oA}\times\Omega_{\oB}$ be a map. Then we say that the lifted map $\varphi^{\uparrow}: \SProb(\Omega_{\oA}\times\Omega_{\oB}) \to \SProb(\Omega_{\oA}\times\Omega_{\oB})$ is a lifted symmetry of $W$ if $\varphi^\uparrow: W(K) \to W(K)$.
\end{definition}
Lifted symmetries are symmetries that arise from the phase space $\Omega_{\oA} \times \Omega_{\oB}$. As we will see later in Section \ref{sec:boxworld}, not every map $\varphi: \Omega_{\oA}\times\Omega_{\oB} \to \Omega_{\oA}\times\Omega_{\oB}$ gives rise to a lifted symmetry, this is because there can be $x \in K$ such that $\varphi^\uparrow(W(x)) \notin W(K)$. Now we proceed to investigate the other type of symmetries, ones that originate from the state space $K$.
\begin{definition}
Let $K$ be a state space and let $\oA$ and $\oB$ be observables. Let $W$ be a Wigner representation of $K$ with respect to $\oA$ and $\oB$ and let $\Psi: \SProb(\Omega_{\oA} \times \Omega_{\oB}) \to \SProb(\Omega_{\oA} \times \Omega_{\oB})$ be a symmetry of $W$. We say that $\Psi$ is a transported symmetry if there is a channel $\Phi: K \to K$ such that $\Psi \circ W = W \circ \Phi$, where $\circ$ denotes the composition of maps. That is, $\Psi$ is a transported symmetry if the following diagram commutes:
\begin{equation} \label{eq:symmetry-transported-diagram}
\begin{tikzcd}
K \arrow[d, "\Phi"] \arrow[r, "W"] & \SProb(\Omega_{\oA} \times \Omega_{\oB}) \arrow[d, "\Psi"] \\
K \arrow[r, "W"] & \SProb(\Omega_{\oA} \times \Omega_{\oB})                
\end{tikzcd}
\end{equation}
\end{definition}

\begin{proposition} \label{prop:symmetry-fauthfulAllTransported}
Let $W$ be faithful Wigner representation of $K$ with respect to $\oA$ and $\oB$. Then every symmetry of $W$ is transported symmetry.
\end{proposition}
\begin{proof}
If $W$ is faithful, then there is a well-defined map $W^{-1}: W(K) \to K$ such that $W^{-1} \circ W = \id_K$. Let $\Psi: \SProb(\Omega_{\oA} \times \Omega_{\oB}) \to \SProb(\Omega_{\oA} \times \Omega_{\oB})$ be a symmetry of $W$, then $\Phi = W^{-1} \circ \Psi \circ W: K \to K$ is a channel. We have $W \circ \Phi = W \circ W^{-1} \circ \Psi \circ W = \Psi \circ W$ and so \eqref{eq:symmetry-transported-diagram} commutes.
\end{proof}

If $W$ is not faithful, then analogical results do not hold, as one can find channels $\Phi: K \to K$ with no corresponding transported symmetry $\Psi: \SProb(\Omega_{\oA} \times \Omega_{\oB}) \to \SProb(\Omega_{\oA} \times \Omega_{\oB})$, as well as symmetries $\Lambda: \SProb(\Omega_{\oA} \times \Omega_{\oB}) \to \SProb(\Omega_{\oA} \times \Omega_{\oB})$ with no corresponding channel $\Phi_\Lambda: K \to K$. This will be demonstrated in the following example.
\begin{example}[Trit to bit representation]
Consider the following example: let $S_3$ be a simplex, that is $S_3 = \conv(\{s_0, s_1, s_2\})$ where $s_0, s_1, s_2$ are affinely independent points. Let $f \in E(S_3)$ be given as $f(s_0) = 1$ and $f(s_1) = f(s_2) = 0$ and consider the observable $\oA$, $\Omega_{\oA} = \{0, 1\}$ and $\mu_{\oA}(s) = (f(s), 1-f(s))$ for $s \in S_3$. Then according to Proposition \ref{prop:Wigner-unique} there is a unique Wigner representation of $S_3$ with respect to $\oA$, given by the map $W: S_3 \to \Prob(\Omega_{\oA})$, $W(s) = \mu_{\oA}(s)$ for all $s \in S_3$. Consider the channel $\Phi_R: S_3 \to S_3$ given as $\Phi(s_0) = s_1$, $\Phi(s_1) = s_2$, $\Phi(s_2) = s_0$, and assume that there is a corresponding transported symmetry $\Psi: \Prob(\Omega_{\oA}) \to \Prob(\Omega_{\oA})$ such that \eqref{eq:symmetry-transported-diagram} commutes. Then we must have $W(s_0) = W(\Phi(s_2)) = \Psi(W(s_2)) = \Psi(W(s_1)) = W(\Phi(s_1)) = W(s_2)$ which is a contradiction.
\end{example}
The important message here is that we can use the phase space $\Omega_{\oA} \times \Omega_{\oB}$ and the Wigner representation $W$ to work with a transformation $\Lambda: \SProb(\Omega_{\oA} \times \Omega_{\oB}) \to \SProb(\Omega_{\oA} \times \Omega_{\oB})$ that is a symmetry of $W$, but not a transported symmetry, i.e., such that there is no channel $\Phi: K \to K$ for which \eqref{eq:symmetry-transported-diagram} would hold. This means that Wigner representations give us a tool to work with symmetries that have no equivalent in terms of the state space $K$. This may seem counter-intuitive at first, but consider the following: in classical thermodynamics, instead of using the microscopic description and describing exact position and momentum of every particle enclosed in a box, we describe the theory using macroscopic description with only handful of thermodynamic quantities, for example we can take volume, pressure and temperature. We will intuitively compare the macroscopic description to Wigner representation of the microscopic system. Then there are transformations which make sense in the macroscopic description, but not in the microscopic description. An easy example is a transformation that increases the volume or pressure; it is obvious what this transformation does in the macroscopic description, but there is not a unique corresponding microscopic transformation.

\section{Groups of transformations} \label{sec:groups}
More interesting than one transformation is a group of transformations. So let $G$ be a finite group, one can now either consider the case when $G$ has an action on $K$, or on $\Omega_{\oA} \times \Omega_{\oB}$, or on $\SProb(\Omega_{\oA} \times \Omega_{\oB})$. The first case is captured by the following definition.

\begin{definition}
Let $K$ be a state space and let $G$ be a group. We say that $G$ has action $\Phi$ on $K$ if for every $g \in G$ there is a channel $\Phi_g: K \to K$ such that for $g_1, g_2 \in K$ we have $\Phi_{g_1 g_2} = \Phi_{g_1} \circ \Phi_{g_2}$.
\end{definition}

In the second case we can lift the respective maps to $\SProb(\Omega_{\oA} \times \Omega_{\oB})$ because we want to consider the cases when the action of $G$ is derived from transforming the values of $\oA$ and $\oB$. $G$ has a right action on $\Omega_{\oA} \times \Omega_{\oB}$ if for every $g \in G$ there is a map $\varphi_g: \Omega_{\oA} \times \Omega_{\oB} \to \Omega_{\oA} \times \Omega_{\oB}$ such that $\varphi_g$ is an automorphism of $\Omega_{\oA} \times \Omega_{\oB}$, i.e., $\varphi_g$ and is both injective and surjective and for $g_1, g_2 \in G$ we have $\varphi_{g_2 g_1} = \varphi_{g_1} \circ \varphi_{g_2}$. Note that here we assume that $G$ has a right action, because when constructing the lifted maps $\varphi_g^\uparrow$ the action is going to get inverted into a left action.

\begin{definition}
Let $K$ be a state space with observables $\oA$ and $\oB$ and let $W$ be a Wigner representation of $K$ with respect to $\oA$ and $\oB$. Let $G$ be a group with action on $\Omega_{\oA} \times \Omega_{\oB}$. We say that $W$ is $G$-symmetric if for every $g \in G$ the map $\varphi_g^\uparrow$ is a symmetry of $W$.
\end{definition}

\begin{theorem} \label{thm:groups-actionOnK}
Let $K$ be a state space with observables $\oA$ and $\oB$ and let $W$ be a faithful Wigner representation of $K$ with respect to $\oA$ and $\oB$. Let $G$ be a group with action on $\Omega_{\oA} \times \Omega_{\oB}$ such that $W$ is $G$-symmetric, then $G$ has an action on $K$.
\end{theorem}
\begin{proof}
The construction is as follows: since $W$ is faithful, $W^{-1}: W(K) \to K$ is well-defined and $W^{-1} \circ W = \id_K$. Let $g \in G$ and define $\Phi_g: K \to K$ as $\Phi_g = W^{-1} \circ \varphi_g^\uparrow \circ W$. It is obvious that $\Phi_g$ is affine, hence a channel, and for $g_1, g_2 \in G$ we have $\Phi_{g_2} \circ \Phi_{g_1} = W^{-1} \circ \varphi_{g_2}^\uparrow \circ W \circ W^{-1} \circ \varphi_{g_1}^\uparrow \circ W = W^{-1} \circ \varphi_{g_1 g_2}^\uparrow \circ W = \Phi_{g_1 g_2}$ and so $g \mapsto \Phi_g$ is an action of $G$ on $K$.
\end{proof}
This shows an interesting property of Wigner representations: faithful, $G$-symmetric Wigner representation induces an action of $G$ on $K$. On the other hand, if we want to have a way of transforming the values of $\oA$ and $\oB$ using the group $G$ and no suitable action of $G$ on $K$ exists, then we can still try to find a non-faithful $G$-symmetric Wigner representation $W$.

\section{Uniqueness of the Wigner representation} \label{sec:unique}
It is known that for the Wigner representation corresponding to position and momentum given by \eqref{eq:Wigner-PQwigner} one can find properties that uniquely specify this Wigner representations among all other possible Wigner representations \cite{Wigner-uniqness, OConnellWigner-uniqness, Gross-discreteWigner}. \update{Specifically for the Wigner functions in finite-dimensional systems, it has been shown that covariance under symmetry transformations given by Clifford unitaries is a defining property of Wigner functions \cite{Gross-discreteWigner}. To a large extent, this precise specification of Wigner functions for finite-dimensional position and momentum obtained in \cite{Gross-discreteWigner, Wootters-discreteWigner} builds on assuming an algebraic (finite) field structure for the discrete phase space. For concreteness, in the simplest setting of a single particle, the phase space is given by the finite field $\mathbb{Z}_d \times \mathbb{Z}_d$ where the factors represent position and momentum, respectively. In recent works, this point of view has led to a more comprehensive insight in the mechanisms that are specific to finite quantum mechanics \cite{Marchiolli2019, Ruzzi2006}.} 

\update{Our operational approach allows us to discover the role of covariance of Wigner representations in more general setting without any topological or geometrical assumptions on the phase space.
We will show that the uniqueness of the Wigner representation can be deduced by a set of assumptions on the observables.} %\update{and incorporates the general notion of symmetries of Wigner representations in operational theories}. 
To do so, we will have to introduce several additional notions, mainly complementary observables. We will also require that for the given observables $\oA$ and $\oB$ the maps $\mu_{\oA}$ and $\mu_{\oB}$ are surjective; this property was recently investigated in \cite{BuscemiKobayashiMinagawa-sharpness} as a form of sharpness of measurements.

 We will although start by introducing translation group $T_{\oA}$.
\begin{definition}
Let $\oA$ be an observable on a state space $K$. We denote by $T_{\oA}$ the translation group on $\Omega_{\oA}$, which is a group with action on $\Omega_{\oA}$ such that for every $a_1, a_2 \in \Omega_{\oA}$ there is some $g \in T_{\oA}$ such that $a_2 = g a_1$.
\end{definition}
The physical interpretation of the group $T_{\oA}$ is clear: it is the group that permutes the outcomes of the observable $\oA$. This is to be seen as adaptation of the group of position translations and momentum translations that play a crucial role in the proofs presented in \cite{Wigner-uniqness, OConnellWigner-uniqness}. Moreover we will use $T_{\oA} \times T_{\oB}$ to denote the Cartesian product of the respective translation groups.

Second we need to define complementary observables. Our definition is in line with \cite{FarkasKaniewskiNayak-MUM}; the intuition is that we call $\oA$ and $\oB$ complementary if for some state $x \in K$ some outcome of observable $\oA$ occurs with certainty, then $\mu_{\oB}(x)$ is completely mixed, i.e., then $\oB$ returns a random outcome. This is again reminiscent of position and momentum in quantum theory as the formal eigenstates of position have infinite uncertainty in momentum. We remind the reader that we use $\tau_{\oA} \in \Prob(\Omega_{\oA})$ to denote the constant probability distribution that is equal mixture of all of the extreme points of $\Prob(\Omega_{\oA})$.
\begin{definition}
Let $K$ be a state space and let $\oA$, $\oB$ be two observables on $K$. We say $\oA$ and $\oB$ are complementary if for every state $x \in K$ we have that whenever $\mu_{\oA}(x) \in \ext(\Prob(\Omega_{\oA}))$ then $\mu_{\oB}(x) = \tau_{\oB}$ and whenever $\mu_{\oA}(x) \in \ext(\Prob(\Omega_{\oA}))$ then $\mu_{\oB}(x) = \tau_{\oB}$.
\end{definition}

We now get the following result:
\begin{theorem} \label{thm:uniqueness}
Let $K$ be a state space, let $\oA$ and $\oB$ be jointly info-complete and complementary observables and assume that $\mu_{\oA}$ and $\mu_{\oB}$ are surjective. Then there is at most one Wigner representation $W$ of $\oA$ and $\oB$ such that for every $(g_1, g_2) \in T_{\oA} \times T_{\oB}$ the map $(g_1, g_2)^{\uparrow}: \SProb(\Omega_{\oA} \times \Omega_{\oB}) \to \SProb(\Omega_{\oA} \times \Omega_{\oB})$ is a lifted symmetry of $W$.
\end{theorem}
\begin{proof}
Since $\oA$ and $\oB$ are jointly info-complete it follows from Proposition~\ref{prop:Wigner-allFaithful} that any Wigner representation $W$ must be faithful. Then since $(g_1, g_2)^{\uparrow}: \SProb(\Omega_{\oA} \times \Omega_{\oB}) \to \SProb(\Omega_{\oA} \times \Omega_{\oB})$ is a lifted symmetry of $W$, as a result of Theorem~\ref{thm:groups-actionOnK} we get that the group $T_{\oA} \times T_{\oB}$ has an action $\Phi$ on $K$, moreover this action must satisfy $\mu_{\oA}(\Phi_{(g_1, g_2)}(x)) = (f_{a, \oA} (\Phi_{(g_1, g_2)}(x)))_{a \in \Omega_{\oA}} = (f_{g_1^{-1}(a), \oA} (x))_{a \in \Omega_{\oA}}$ and $\mu_{\oB}(\Phi_{(g_1, g_2)}(x)) = (f_{b, \oB} (\Phi_{(g_1, g_2)}(x)))_{b \in \Omega_{\oB}} = (f_{g_2^{-1}(b), \oB} (x))_{b \in \Omega_{\oA}}$, where $f_{a,\oA}, f_{b, \oB} \in E(K)$ are the respective functions corresponding to the maps $\mu_{\oA}$ and $\mu_{\oB}$.

Let $e_{\oA} \in T_{\oA}$ and $e_{\oB} \in T_{\oB}$ be the identity elements of the respective groups. We will denote $\fix(T_{\oA})$ and $\fix(T_{\oB})$ the sets of fixed points of all maps of the form $\Phi_{(g_1, e_{\oB})}$ and $\Phi_{(e_{\oA}, g_2)}$ respectively, i.e., $\fix(T_{\oA}) = \{ x \in K: \Phi_{(g_1, e_{\oB})}(x) = x, \forall g_1 \in T_{\oA} \}$. The proof works in two steps: we will first prove that $K \subset \aff(\fix(T_{\oA}) \cup \fix(T_{\oB}))$ and then we will use this result to show that then there is at most one covariant Wigner representation of $\oA$ and $\oB$.

In order to prove that $K \subset \aff(\fix(T_{\oA}) \cup \fix(T_{\oB}))$ we just need to find $x_1 \in \fix(T_{\oA})$ and $x_2 \in \fix(T_{\oB})$ such that $x = \alpha x_1 + (1-\alpha) x_2$ for some $\alpha \in \RR$. Since $\oA$, $\oB$ are surjective and complementary it follows that we can always find $x_1, x_2 \in K$ such that $\mu_{\oA}(x_1) = \tau_{\oA}$ and $\mu_{\oB}(x_1)$ is an arbitrary probability distribution and such that $\mu_{\oA}(x_2)$ is an arbitrary probability distribution and $\mu_{\oB}(x_2) = \tau_{\oB}$; this is straightforward to see as for example in case of $x_1$ we know that since $\oA$ is surjective there are states that get mapped to $\ext(\Prob(\Omega_{\oB}))$ and we can construct $x_1$ as the respective convex combinations of such states, $\mu_{\oA}(x_1) = \tau_{\oA}$ follows from complementarity of $\oA$ and $\oB$. Moreover due to joint info-completness of $\oA$ and $\oB$ we get that $x_1$ is uniquely specified by $\mu_{\oB}(x_1)$ and $x_2$ is uniquely specified by $\mu_{\oA}(x_2)$. Additionally it is straightforward to check that we have $x_1 \in \fix(T_{\oA})$ and $x_2 \in \fix(T_{\oB})$. So let $x \in K$ and assume that $\mu_{\oA}(x) \neq \tau_{\oA}$ and $\mu_{\oB}(x) \neq \tau_{\oB}$ since in these cases the proof is trivial. We want to choose $x_1 \in \fix(T_{\oA})$ and $x_2 \in \fix(T_{\oB})$ such that for some $\alpha \in \RR \setminus \{0, 1\}$ we have $\mu_{\oA}(x_2) = \frac{\mu_{\oA}(x) - \alpha \tau_{\oA}}{1-\alpha}$ and $\mu_{\oB}(x_1) = \frac{\mu_{\oB}(x) - (1-\alpha) \tau_{\oB}}{\alpha}$. This may seem easy at first, but we also have to enforce that $\mu_{\oA}(x_2) \in \Prob(\Omega_{\oA})$ and $\mu_{\oB}(x_1) \in \Prob(\Omega_{\oB})$. Let us, without the loss of generality, pick $\alpha < 0$, then we have $\frac{\mu_{\oA}(x) - \alpha \tau_{\oA}}{1-\alpha} = \frac{\mu_{\oA}(x) + \abs{\alpha} \tau_{\oA}}{1+\abs{\alpha}}$ and so $\mu_{\oA}(x_2) \in \Prob(\Omega_{\oA})$ follows. We also have $\frac{\mu_{\oB}(x) - (1-\alpha) \tau_{\oB}}{\alpha} = \frac{(1+\abs{\alpha}) \tau_{\oB} - \mu_{\oB}(x)}{\abs{\alpha}}$ which can be picked such that $\mu_{\oB}(x_1) \in \Prob(\Omega_{\oB})$ by choosing $\abs{\alpha}$ large enough. We thus have $x_1, x_2 \in K$ with the wanted properties. We now observe that we have $\mu_{\oA}(x) = \alpha \mu_{\oA}(x_1) + (1-\alpha) \mu_{\oA}(x_2)$ and $\mu_{\oB}(x) = \alpha \mu_{\oB}(x_1) + (1-\alpha) \mu_{\oB}(x_2)$, from joint info-completeness of $\oA$ and $\oB$ it follows that $x = \alpha x_1 + (1-\alpha) x_2$ as intended. We have thus proved that $K \subset \aff(\fix(T_{\oA}) \cup \fix(T_{\oB}))$.

Let us assume that $W_1, W_2: K \to \SProb(\Omega_{\oA} \times \Omega_{\oB})$ are two Wigner representations of $K$ corresponding to $\oA$ and $\oB$ for which for every $(g_1, g_2) \in T_{\oA} \times T_{\oB}$ the map $(g_1, g_2)^{\uparrow}: \SProb(\Omega_{\oA} \times \Omega_{\oB}) \to \SProb(\Omega_{\oA} \times \Omega_{\oB})$ is a lifted symmetry. Let $x_1 \in \fix(T_{\oA})$, we then have that $(g, e_{\oB})^{\uparrow}(W_1(x_1)) = W_1(x_1)$ and $(g, e_{\oB})^{\uparrow}(W_2(x_1)) = W_2(x_1)$ for any $g \in T_{\oA}$. We will now proceed by showing that $W_1(x_1) = W_2(x_1)$, so let us denote $(X_{ab})_{a \in \Omega_{\oA}, b \in \Omega_{\oB}} = W_1(x_1) - W_2(x_1)$, here $X_{ab}$ is the element of matrix that we get when we represent $W_1(x_1)$ and $W_2(x_1)$ as matrices and subtract them from each other. One can then easily prove that since $\int_{\Omega_{\oA}} W_1(x) \dd a = \int_{\Omega_{\oA}} W_2(x) \dd a$ and $\int_{\Omega_{\oB}} W_1(x) \dd b = \int_{\Omega_{\oB}} W_2(x) \dd b$ we get $\sum_{a \in \Omega_{\oA}} X_{ab} = \sum_{b \in \Omega_{\oB}} X_{ab} = 0$. Moreover from $(g, e_{\oB})^{\uparrow}(W_1(x)) = W_1(x)$ and $(g, e_{\oB})^{\uparrow}(W_2(x)) = W_2(x)$ we get $X_{ab} = X_{(g a) b}$ for all $g \in T_{\oA}$ which implies that $X_{a_1 b} = X_{a_2 b}$ for all $a_1, a_2 \in \Omega_{\oA}$ and $b \in \Omega_{\oB}$. So let $\alpha \in \Omega_{\oA}$ be any fixed element, then we get $0 = \sum_{a \in \Omega_{\oA}} X_{ab} = \abs{\Omega_{\oA}} X_{\alpha b}$ and thus $X_{\alpha b} = 0$ for all $b \in \Omega_{\oB}$. We thus have $W_1(x_1) = W_2(x_1)$ for all $x_1 \in \fix(T_{\oA})$.

One can, using analogical proof, show that for any $x_2 \in \fix(T_{\oB})$ we have $W_1(x_2) = W_2(x_2)$. This implies that we must have $W_1(x) = W_2(x)$ for any $x \in \aff(\fix(T_{\oA}) \cup \fix(T_{\oB}))$, but since, as we already proved, $K \subset \aff(\fix(T_{\oA}) \cup \fix(T_{\oB}))$, we get that $W_1(x) = W_2(x)$ for all $x \in K$.
\end{proof}

\tao{ \textbf{Remark:} This Wigner representation for which every permutation of the observable outcomes corresponds to a lifted symmetry is called a covariant Wigner representation. Theorem \ref{thm:uniqueness} states that a state space can admit at most one covariant Wigner representation for a pair of jointly info-complete, complementary and surjective observables. 
The statement is analog to the uniqueness for position and 
momentum observables \cite{OConnellWigner-uniqness} and the proof here essentially relies on the convex structure of the underlying state space and no state construction on Hilbert spaces has to be assumed. \update{Also the logic of our proof allows for a simple generalisation to Wigner representations of more than two observables for which the notions of surjectivity and joint info-completeness can be generalised in a straightforward manner.}

Considering the matrix representation $W_{ab}(x)$ of a Wigner representation $W(x)$ with respect to observables $\oA$ and $\oB$, covariance describes a memorable feature of those matrix elements. Let $g_1 \in T_{A}$ and $g_{2} \in T_{B}$ be any permutations of the elements in $\Omega_{A}$ and $\Omega_{B}$ and let $\Phi_{(g_{1},g_{2})}: K \rightarrow K$ be the corresponding channel describing the state transformation. Then, $W$ is covariant if for all $x\in K, a\in \Omega_{A}$ and $b \in \Omega_{B}$ it holds true that 
\begin{equation} \label{eq:covariance_matrix_elements}
W(\Phi_{(g_1,g_2)}(x))_{a b} = W(x)_{g_{1}^{-1}(a) g_{2}^{-1}(b)}.     
\end{equation}
In words, the covariant Wigner representation $W$ transforms in a suitable way under a permutation of outcomes. While Thm. \ref{thm:uniqueness} provides a uniqueness result for covariant Wigner representations, one might ask if an existence is always guaranteed. The answer to this question is negative. To see that, consider the situation in which the observables are kept fixed while parts of the state space are removed in a way such that convexity is preserved. This necessarily makes some of the channels $\Phi_{(g_{1},g_{2})}$ ill-defined so that not every permutation of observable outcomes can correspond to a lifted symmetry. The phenomenon of non-existence of covariant representations will also be illustrated within the next section when the covariant Wigner representation of a real quantum bit is discussed. 
}

\section{Example: Quantum theory} \label{sec:QT}
We will first review the Wigner representation of finite-dimensional quantum theory. Our presentation is based on \cite{Wootters-discreteWigner}. Let $\Ha$ be a complex Hilbert space such that $\dH = d$. Further, let $\bound_h(\Ha)$ denote the set of hermitian operators on $\Ha$. $A \in \bound_h(\Ha)$ is positive semidefinite if $\bra{\psi}A\ket{\psi} \geq 0$ for all $\ket{\psi} \in \Ha$. Let $\Tr(A)$ denote the trace of $A$. We will use $\sigma_x$, $\sigma_y$ and $\sigma_z$ to denote the respective Pauli matrices.

The state space in quantum theory is the set of density matrices $\dens(\Ha) = \{ \rho \in \bound_h(\Ha): \rho \geq 0, \Tr(\rho) = 1 \}$ and the corresponding effect algebra is $\effect(\Ha) = \{ E \in \bound_h(\Ha): 0 \leq E \leq \I \}$, so that every measurement is given by a POVM, that is by a tuple of operators $M_1, \ldots, M_n$ such that $M_i \in \effect(\Ha)$ for all $i \in \{1, \ldots, n\}$ and $\sum_{i=1}^n M_i = \I$. Every POVM $M = \{M_i\}$ induces an observable $q_M$ in the general sense via the map $\mu_M:\dens(\Ha) \rightarrow \Prob(n): \rho \mapsto \{\Tr(M_{i} \rho)\}_{i=1}^{n}$ and the outcome set $\Omega_{M} = \{1,\dots,n\}$. A simple example for a POVM is given by the set $\{\ket{a_i}\bra{a_i}\}$ when $\ket{a_i}$ is an orthonormal basis of $\Ha$.

\subsection{Binary position and momentum}
We start by considering qubits for which $d=2$.
Let $\oA$ and $\oB$ be observables given as follows: let $\Omega_{\oA} = \Omega_{\oB} = \Omega = \{ 0, 1\}$ and let $\{\frac{1}{2}(\I + \sigma_z), \frac{1}{2}(\I - \sigma_z)\}$ be the POVM corresponding to $\mu_{\oA}$ and $\{\frac{1}{2}(\I + \sigma_x), \frac{1}{2}(\I - \sigma_x)\}$ be the POVM corresponding to $\mu_{\oB}$. In other words, let $\rho \in \dens(\Ha)$, then we have $\mu_{\oA}(\rho) = ( \frac{1}{2}(1 + \Tr(\rho \sigma_z)), \frac{1}{2}(1 - \Tr(\rho \sigma_z)) )$, $\mu_{\oB}(\rho) = ( \frac{1}{2}(1 + \Tr(\rho \sigma_x)), \frac{1}{2}(\Tr(1 - \Tr(\rho \sigma_x)) )$; these observables are complementary and can hence be seen as a version of binary position and momentum observables. On the other hand, they are not jointly info-complete. 
We can now define one of the many possible Wigner representation of $\dens(\Ha)$ with respect to $\oA$, $\oB$ using the operators $W_{ij} \in B_h(\Ha)$ defined for $i,j \in \Omega$ as follows: $W_{00} = \frac{1}{4} ( \sigma_x + \sigma_y + \sigma_z + \I )$, $W_{01} = \frac{1}{4} ( - \sigma_x - \sigma_y + \sigma_z + \I )$, $W_{10} = \frac{1}{4} ( \sigma_x - \sigma_y - \sigma_z + \I )$, $W_{11} = \frac{1}{4} ( - \sigma_x + \sigma_y - \sigma_z + \I )$. Then
\begin{equation} \label{eq:quantum_wigner}
W(\rho) =
\begin{pmatrix}
\Tr(\rho W_{00}) & \Tr(\rho W_{01}) \\
\Tr(\rho W_{10}) & \Tr(\rho W_{11})
\end{pmatrix},
\end{equation}
where we used a matrix to represent the product structure of $\Omega \times \Omega$. We have $\int_{\Omega_{\oB}} W(\rho) \dd b = ( \Tr(\rho (W_{00} + W_{01})), \Tr(\rho (W_{10} + W_{11})) )$, $\int_{\Omega_{\oA}} W(\rho) \dd a = ( \Tr(\rho (W_{00} + W_{10})), \Tr(\rho (W_{01} + W_{11})) )$, and one can easily verify that $\oA$ and $\oB$ are marginals of $W$, i.e., that \eqref{eq:Wigner-marginalA} and \eqref{eq:Wigner-marginalB} are satisfied.

The Wigner representation $W$ is not positive: one can either verify that the measurements $\mu_{\oA}$ and $\mu_{\oB}$ are not compatible, this is in fact straightforward to do as the corresponding POVMs are projection-valued measures and for projection-valued measures compatibility is equivalent with commutativity \cite{HeinosaariReitznerStano-compatibility}. But to provide a direct proof, let $\rho_n = \frac{1}{2} ( \I + \dfrac{1}{\sqrt3} (\sigma_x + \sigma_y + \sigma_z) )$ then $\Tr(\rho_n W_{00}) = \frac{1 - \sqrt{3}}{2} \approx -0,366 < 0$. The Wigner representation $W$ is faithful; observe that we have
\begin{align}
W_{00} + W_{10} - ( W_{01} + W_{11} ) &= \sigma_x, \label{eq:QT-Wsigmax} \\
W_{00} + W_{11} - ( W_{01} + W_{10} ) &= \sigma_y, \label{eq:QT-Wsigmay} \\
W_{00} + W_{01} - ( W_{10} + W_{11} ) &= \sigma_z. \label{eq:QT-Wsigmaz}
\end{align}
Let $\rho_1, \rho_2 \in \dens(\Ha)$, then $W(\rho_1) = W(\rho_2)$ implies $\Tr(\rho_1 W_{ij}) = \Tr(\rho_2 W_{ij})$ for all $i,j \in \Omega$. Using \eqref{eq:QT-Wsigmax} - \eqref{eq:QT-Wsigmaz} we get $\Tr(\rho_1 \sigma_x) = \Tr(\rho_2 \sigma_x)$, $\Tr(\rho_1 \sigma_y) = \Tr(\rho_2 \sigma_y)$, and $\Tr(\rho_1 \sigma_z) = \Tr(\rho_2 \sigma_z)$. It follows that $\rho_1 = \rho_2$ and so $W$ is faithful.

Now we examine the symmetry properties of this Wigner representation. Let $T_4$ be the group of permutations of the phase space $\Omega \times \Omega = \{ 00, 01, 10, 11 \}$. One can easily see that this group is generated by the elements $\varphi_{01}$, $\varphi_{10}$, and $\varphi_{11}$, where $\varphi_{ij}$ is defined as the elements that swap $00$ with $ij$, for $i,j \in \Omega$, e.g., $\varphi_{01}(00) = 01$, 
$\varphi_{01}(01) = 00$, $\varphi_{01}(10) = 10$, $\varphi_{01}(11) = 11$. To show that the lifted map $\varphi_{01}^\uparrow$ is a symmetry of $W$, let $\rho \in \dens(\Ha)$ be given as $\rho = \frac{1}{2}(\I + r_x \sigma_x + r_y \sigma_y + r_z \sigma_z)$ for some $r_x, r_y, r_z \in \RR$ such that $r_x^2 + r_y^2 + r_z^2 \leq 1$. Let $\rho'$ be given as $\rho' = \frac{1}{2}(\I - r_y \sigma_x - r_x \sigma_y + r_z \sigma_z)$, we then have $\varphi_{01}^\uparrow (W(\rho)) = W(\rho')$ and so $\varphi_{01}^\uparrow$ is a symmetry of $W$. Moreover, by construction, $\varphi_{01}^\uparrow$ is a lifted symmetry. To show that also $\varphi_{10}^\uparrow$ is a symmetry of $W$, take $\rho' = \frac{1}{2}(\I + r_x \sigma_x - r_z \sigma_y - r_y \sigma_z)$ and to show that $\varphi_{11}^\uparrow$ is a symmetry of $W$, take $\rho' = \frac{1}{2}(\I - r_z \sigma_x + r_y \sigma_y - r_x \sigma_z)$. It then follows that $W$ is $T_4$-symmetric, since $\varphi_{01}$, $\varphi_{10}$ and $\varphi_{11}$ are generators of $T_4$. \update{Importantly, the symmetry of $W$ is no contradiction to the result that there exist no Clifford-covariant Wigner function in finite-dimensional quantum systems \cite{Zhu2016, Raussendorf2023}. The set of allowed symmetry transformations on states is less restrictive than in the usual description for quantum systems by allowing all affine linear transformations. Concretely, the transformation $\rho \mapsto \rho'$ is a reflection in the Bloch ball and hence not unitary.}

Let us further compare our results with the introduction of the Wigner representation $W$ in \cite{Wootters-discreteWigner}. It was required in \cite[Property (A3')]{Wootters-discreteWigner} that not only the marginals $\int_{\Omega_{\oA}} W(\rho) \dd a$ and $\int_{\Omega_{\oB}} W(\rho) \dd b$ are positive, but also that the sum along the diagonals is positive, that is that also $\Tr(\rho(W_{00}+W_{11}) \geq 0$ and $\Tr(\rho(W_{01}+W_{10}) \geq 0$. One can see that these conditions follow from that $W$ is $T_4$-symmetric.

One might now ask what happens if one restricts the Wigner representation $W$
to only act on the $XZ$-plane, i.e. on the state space $D(\Ha)_{XZ}=\{\rho \in D(\Ha) : \Tr(\sigma_y \rho)=0\}$ of $2\times2-$ real symmetric matrices, as the two observables $\oA$ and $\oB$ are jointly info-complete in that case. By Theorem~\ref{thm:uniqueness}, we know that there can be at most one Wigner representation on $D(\Ha)_{XZ}$ of $\oA$ and $\oB$ such that $g^{\uparrow}$ is a lifted symmetry of it for all $g\in T_4$. It turns out that $W$ defined in \eqref{eq:quantum_wigner} is the unique covariant Wigner representation. To see this, first notice that on $D(\Ha)_{XZ}$, $W$ can be written as  
\begin{equation} \label{eq:reduced_quantum_wigner}
W(\rho) = \frac{1}{4}
\begin{pmatrix}
\Tr(\rho (\sigma_{x} + \sigma_{z} + \I)) & \Tr(\rho (- \sigma_{x} + \sigma_{z} + \I)) \\
\Tr(\rho (\sigma_{x} - \sigma_{z} + \I)) & \Tr(\rho (- \sigma_{x} - \sigma_{z} + \I))
\end{pmatrix}.
\end{equation}
To verify the covariance of $W$, it can be seen that the channels $\Phi_{01}$ and $\Phi_{10}$ on $D(\Ha)_{XZ}$ that correspond to the permutations $\phi_{01}$ and $\phi_{10}$ are given by reflections across the $X$-axis and $Z$-axis, respectively. A reflection across the $X$-axis corresponds to the exchange of the values of $\Tr(\rho \sigma_{x})$ and $-\Tr(\rho \sigma_{x})$ and a reflection across the $z$-axis corresponds to the exchange of $\Tr(\rho \sigma_{z})$ and $-\Tr(\rho \sigma_{z})$. With this insight and the definition of covariance in terms of matrix elements \eqref{eq:covariance_matrix_elements}, one can readily verify that the Wigner representation \eqref{eq:reduced_quantum_wigner} is covariant.

In the case of quantum theory one can invert the logic and look for Wigner representations that are $G$-symmetric with respect to a given group $G$. This situation is analyzed in \cite{Gross-discreteWigner}, where it is shown that covariance with respect to Clifford group essentially fixes the Wigner representation $W$.

By deformation of the underlying state space one can construct an example where no covariant representation can exist. Consider the artificial state space $K = D(\Ha)_{XZ}\setminus \{\rho: \Tr(\sigma_{z} \rho) \geq 0, \Tr(\sigma_{z} \rho) + \abs{\Tr(\sigma_{x} \rho)} \geq 1\}$ 
and the same observables $\oA$ and $\oB$ as above. Clearly, the assumptions of Theorem~\ref{thm:uniqueness} of joint info-completeness, complementarity and surjectivity are still satisfied. However, the channel $\Phi_{01}$ that corresponds to a reflection around the $x$-axis is not well-defined on the whole state space $K$ anymore. To see this, consider states $\rho \in K$ for which $-\Tr(\rho \sigma_{z}) + \abs{\Tr(\sigma_{x} \rho)} \geq 1$ which would be mapped to points outside of $K$. Therefore, by the non-existence of channels that permute the outcome probabilities of states, there can not be a covariant Wigner representation of $K$ with respect to $\oA$ and $\oB$.

\update{
\subsection{Position and momentum of qudits}

The example of binary position and momentum of a qubit can be generalized to quantum mechanics on larger discrete phase spaces of $d \times d$ lattice sites. To this end, we take $\mathcal{H}$ to be $d$-dimensional spanned by the orthonormal basis $\{\ket{0},\dots,\ket{d-1}\}$. This basis is naturally associaticed to a position observable $q = \sum_{j=0}^{d-1} = j \ket{j}\bra{j}$. The corresponding momentum observable is $p = \sum_{j=0}^{d-1} j \ket{\Tilde{j}}\bra{\Tilde{j}}$ where $\{\ket{\Tilde{j}}\}_{j=0}^{d-1}$ is the Fourier basis $\ket{\Tilde{j}} = \sum_{l=0}^{d-1} e^{\frac{2\pi i\cdot j \cdot l}{d}} \ket{l}$. Taking a qudit state $\rho$ and invoking our view on observables as probability assignments $\mu_{q}$ and $\mu_{p}$ we have $\mu_{q}(\rho) = (\bra{j}\rho \ket{j})_{j}$ and $\mu_{p}(\rho) = (\bra{\Tilde{j}}\rho \ket{\Tilde{j}})_{j}$. To apply Theorem~\ref{thm:uniqueness}, the state space has be one such that the observables $p$ and $q$ are informationally complete and this condition is clearly not fulfilled when choosing the set $K$ of all qudit density matrices. To this end, we define the state space $\dens(\Ha)/_{\sim}$ of equivalence classes corresponding to the relation $\rho \sim \rho^{\prime} \Leftrightarrow (\mu_{q}(\rho) = \mu_{q}(\rho^{\prime}) \textnormal{ and } \mu_{p}(\rho) = \mu_{p}(\rho^{\prime}))$. On this state space, the observables $q$ and $p$ are clearly jointly info-complete, complementary and surjective. We can conclude that covariance of the Wigner representation of $q$ and $p$ implies uniqueness in all finite dimensions $d$. Notice again that the symmetry transformations on the states do not have to be unitary and the existence of covariant Wigner representations for even $d$ does not contradict the results in Ref.~\cite{Zhu2016, Raussendorf2023}.

\subsection{Wigner representations of mutually unbiased measurements}
So far, the Wigner representations under consideration corresponded to two observables. Our operational formalism allows for a simple generalisation to Wigner representation of a collection of more than two observables. Such a construction is uselful for studying sets of mutually unbiased bases and their associated observables. Two orthormal bases $\ket{a_i}$ and $\ket{b_i}$ of a $d$-dimensional Hilbert space are called mutually unbiased if $\vert\bra{a_i}b_j\rangle\vert^2 = \frac{1}{d}$ for all $i,j=1,\dots,d$. It is a fact that observables related to POVMs of mutually unbiased bases $A = \{\ket{a_i}\bra{a_i}\}$ and $B = \{\ket{b_i}\bra{b_i}\}$ are complementary. This is due to the fact that the states have deterministic measurement outcomes for $A$ are exactly $\ket{a_i}\bra{a_i}$ and similarly for $B$. A set of $N$ bases is called mutually unbiased if every pair of bases of the set is mutually unbiased. For a comprehensive review on mutually unbiased bases, see Ref.~\cite{DURT2010}.   

One famous set of mutually unbiased bases is given by the eigenvectors of the three Pauli matrices; the corresponding measurements are $M_x = \{\frac{1}{2}(\I + \sigma_x), \frac{1}{2}(\I - \sigma_x)\}$, $M_y = \{\frac{1}{2}(\I + \sigma_y), \frac{1}{2}(\I - \sigma_y)\}$ and $M_z = \{\frac{1}{2}(\I + \sigma_z), \frac{1}{2}(\I - \sigma_z)\}$. A Wigner representation of the qubit state space corresponding to this set of mutually unbiased measurements is given by $W(\rho) = \{\Tr(W_{ijk} \rho)\}_{ijk}$ with the tensor $\{W_{ijk}\}_{ijk=0}^{1}$ defined as
\begin{equation}
    W_{ijk} = \frac{1}{8} ( (-1)^{i}\sigma_x + (-1)^{j} \sigma_y + (-1)^{k}\sigma_z + \I ).
\end{equation}
One can simply check that $W$ gives rise to the correct measurement probabilities, e.g. $\Tr(M_z^{k} \rho) = \sum_{ij} \Tr(W_{ijk} \rho)$ for $k=0,1$. The Wigner representation is also covariant as for any state $\rho$ and any transformation $(i,j,k) \mapsto (\Bar{i},\Bar{j},\Bar{k}) : = (i \oplus s_x,j \oplus s_y,k \oplus s_z)$ with fixed $s_x, s_y, s_z \in \{0,1\}$ there exists a state $\widetilde{\rho}$ such that $\Tr(W_{ijk} \rho) = \Tr(W_{\Bar{i}\Bar{j}\Bar{k}} \widetilde{\rho})$. The state $\widetilde{\rho}$ is explicitly computed as 
\begin{equation}
    \widetilde{\rho} = \frac{1}{2}(\mathds{1} + (-1)^{s_x}r_x \sigma_x + (-1)^{s_y}r_y \sigma_y + (-1)^{s_z}r_z \sigma_z)
\end{equation}
where $(r_x, r_y, r_z)$ is the Bloch vector of $\rho$. By Theorem \ref{thm:uniqueness}, we can state that $W$ is the unique covariant Wigner representation of the set of mutually unbiased measurements given by the Pauli matrices. Also for sets of mutually unbiased measurements in higher dimensions, it can be stated that the existence of a covariant Wigner representation already implies it's uniqueness.}

\begin{figure}
\includegraphics[width=0.7\linewidth]{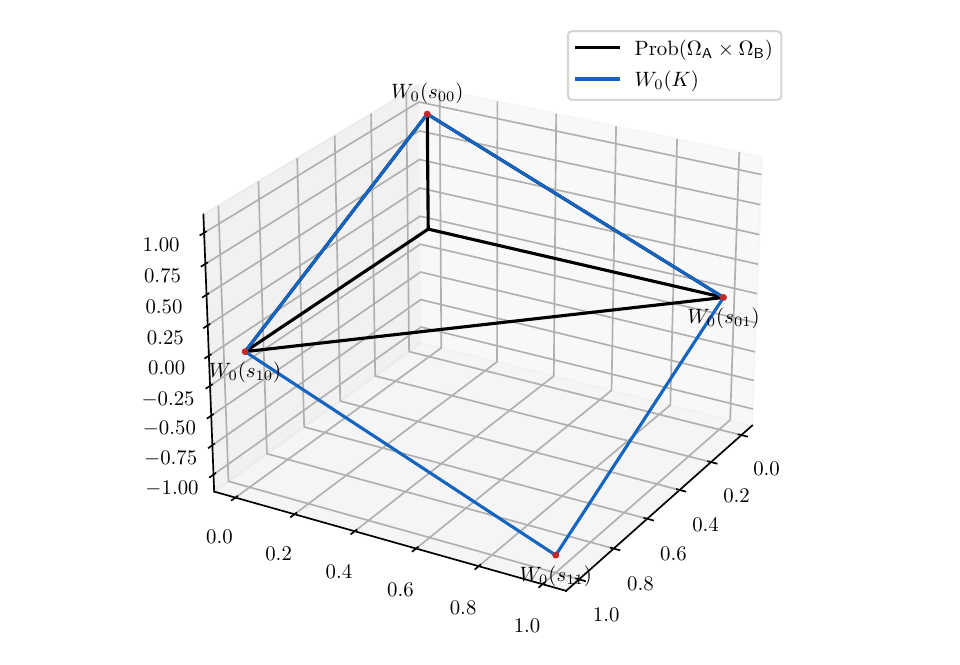}
\caption{The Wigner representation $W_0$ of $S$ with respect to $\oA$ and $\oB$.}
\label{fig:boxworld-W0}
\end{figure}

\section{Example: Boxworld} \label{sec:boxworld}
In this section we will review several Wigner representations of the so-called Boxworld theory. Boxworld was introduced in \cite{Barrett-GPTinformation} and since it is often used as a non-classical toy theory. The state space that we will work with has a state space that is a square $S$ with four different extreme points $s_{00}$, $s_{01}$, $s_{10}$ and $s_{11}$ that satisfy $\frac{1}{2}(s_{00}+s_{11}) = \frac{1}{2}(s_{01}+s_{10})$. Let $f_x, f_y \in E(S)$ be the functions defined as $f_x(s_{00}) = f_x(s_{01}) = 0$, $f_x(s_{10}) = f_x(s_{11}) = 1$, $f_y(s_{00}) = f_y(s_{10}) = 0$, $f_y(s_{01}) = f_y(s_{11}) = 1$. Let $\oA$ and $\oB$ be observables defined as follows: let $\Omega_{\oA} = \Omega_{\oB} = \{0, 1\}$ and for $s \in S$ we have $\mu_{\oA}(s) = ( f_x(s), 1 - f_x(s))$, $\mu_{\oB}(s) = ( f_y(s), 1 - f_y(s))$. It is well known that $\mu_{\oA}$ and $\mu_{\oB}$ are incompatible \cite{BuschHeinosaariSchultzStevens-compatibility,JencovaPlavala-maxInc} and so we can not expect to find a positive Wigner representation of $S$ with respect to $\oA$ and $\oB$. Moreover $\oA$ and $\oB$ are jointly info-complete but not complementary. One can easily deduce that any Wigner representation $W$ of $S$  with respect to $\oA$ and $\oB$ is of the form
\begin{equation} \label{eq:boxworld-Wfxfyq}
W(s) =
\begin{pmatrix}
q(s) & f_x(s) - q(s) \\
f_y(s) - q(s) & 1 + q(s) - (f_x(s) + f_y(s))
\end{pmatrix}
\end{equation}
where $s \in S$ and $q: S \to \RR$ is an affine function. The marginals then are
\begin{align}
\int_{\Omega_{\oB}} W(s) \dd b &= ( q(s) + f_x(s) - q(s), f_y(s) - q(s) + 1 + q(s) - (f_x(s) + f_y(s)) ) = \mu_{\oA}(s), \\
\int_{\Omega_{\oA}} W(s) \dd a &= ( q(s) + f_y(s) - q(s), f_x(s) - q(s) + 1 + q(s) - (f_x(s) + f_y(s)) ) = \mu_{\oB}(s),
\end{align}
as needed. One can see that all of these Wigner representations are faithful as a result of Proposition \ref{prop:Wigner-allFaithful}. In the following, we will use $T_4$ to denote the permutation group of $\Omega_{\oA} \times \Omega_{\oB}$ and we denote
\begin{align}
&\delta_{00} =
\begin{pmatrix}
1 & 0 \\
0 & 0
\end{pmatrix},
&&\delta_{01} =
\begin{pmatrix}
0 & 1 \\
0 & 0
\end{pmatrix}, \\
&\delta_{10} =
\begin{pmatrix}
0 & 0 \\
1 & 0
\end{pmatrix},
&&\delta_{11} =
\begin{pmatrix}
0 & 0 \\
0 & 1
\end{pmatrix}.
\end{align}
Then $\delta_{ij}$, $i,j \in \Omega$, represent the extreme points of $\Prob(\Omega_{\oA}\times\Omega_{\oB})$ and for a given $q$ we can then plot $W(S)$ as a subset of $\aff(\Prob(\Omega_{\oA}\times\Omega_{\oB}))$.

The simplest choice of $q(s) = 0$ for all $s \in S$ yields a Wigner representation $W_0$, given as:
\begin{align}
&W_0(s_{00}) =
\begin{pmatrix}
0 & 0 \\
0 & 1
\end{pmatrix},
&&W_0(s_{01}) =
\begin{pmatrix}
0 & 0 \\
1 & 0
\end{pmatrix}, \\
&W_0(s_{10}) =
\begin{pmatrix}
0 & 1 \\
0 & 0
\end{pmatrix},
&&W_0(s_{11}) =
\begin{pmatrix}
0 & 1 \\
1 & -1
\end{pmatrix}.
\end{align}
One can immediately see that $W_0$ is not positive, i.e., that $W(S) \not\subset \Prob(\Omega_{\oA}\times\Omega_{\oB})$, see Figure \ref{fig:boxworld-W0}. One can also see that the only lifted symmetry of $W_0(S)$ is given by the map $\varphi_{10 \leftrightarrow 01}: \Omega_{\oA}\times\Omega_{\oB} \to \Omega_{\oA}\times\Omega_{\oB}$ that exchanges $01$ and $10$, i.e., $\varphi_{10 \leftrightarrow 01}(00) = 00$, $\varphi_{10 \leftrightarrow 01}(01) = 10$, $\varphi_{10 \leftrightarrow 01}(10) = 01$, $\varphi_{10 \leftrightarrow 01}(11) = 11$, and it follows that $W_0$ is not $T_4$-symmetric.

\begin{figure}
\includegraphics[width=0.7\linewidth]{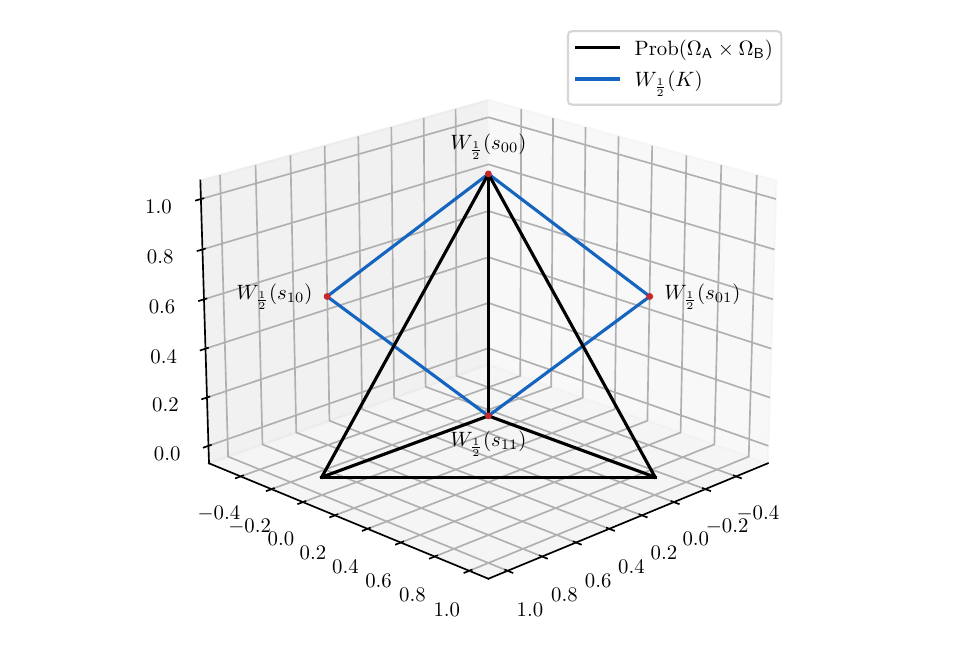}
\caption{The Wigner representation $W_{\frac{1}{2}}$ of $S$ with respect to $\oA$ and $\oB$.}
\label{fig:boxworld-W12}
\end{figure}

Other possible choice is $q = \frac{1}{2}(f_x + f_y)$ which yields the Wigner representation $W_{\frac{1}{2}}$ given as:\begin{align}
&W_{\frac{1}{2}}(s_{00}) =
\begin{pmatrix}
0 & 0 \\
0 & 1
\end{pmatrix},
&&W_{\frac{1}{2}}(s_{01}) = \dfrac{1}{2}
\begin{pmatrix}
1 & -1 \\
1 & 1
\end{pmatrix}, \\
&W_{\frac{1}{2}}(s_{10}) = \dfrac{1}{2}
\begin{pmatrix}
1 & 1 \\
-1 & 1
\end{pmatrix},
&&W_{\frac{1}{2}}(s_{11}) =
\begin{pmatrix}
1 & 0 \\
0 & 0
\end{pmatrix},
\end{align}
see Figure \ref{fig:boxworld-W12}. Let $\varphi_{10 \leftrightarrow 01}, \varphi_{00 \leftrightarrow 11}: \Omega_{\oA}\times\Omega_{\oB} \to \Omega_{\oA}\times\Omega_{\oB}$ be the maps that exchange $10$, $01$, and $00$, $11$, respectively. Then it is easy to verify that $\varphi_{10 \leftrightarrow 01}^\uparrow$ and $\varphi_{00 \leftrightarrow 11}^\uparrow$ are symmetries of $W_{\frac{1}{2}}$, as we have:
\begin{align}
&\varphi_{10 \leftrightarrow 01}^\uparrow(W_{\frac{1}{2}}(s_{00})) = W_{\frac{1}{2}}(s_{00}),
&&\varphi_{00 \leftrightarrow 11}^\uparrow(W_{\frac{1}{2}}(s_{00})) = W_{\frac{1}{2}}(s_{11}), \\
&\varphi_{10 \leftrightarrow 01}^\uparrow(W_{\frac{1}{2}}(s_{01})) = W_{\frac{1}{2}}(s_{10}),
&&\varphi_{00 \leftrightarrow 11}^\uparrow(W_{\frac{1}{2}}(s_{01})) = W_{\frac{1}{2}}(s_{01}), \\
&\varphi_{10 \leftrightarrow 01}^\uparrow(W_{\frac{1}{2}}(s_{10})) = W_{\frac{1}{2}}(s_{01}),
&&\varphi_{00 \leftrightarrow 11}^\uparrow(W_{\frac{1}{2}}(s_{10})) = W_{\frac{1}{2}}(s_{10}), \\
&\varphi_{10 \leftrightarrow 01}^\uparrow(W_{\frac{1}{2}}(s_{11})) = W_{\frac{1}{2}}(s_{11}),
&&\varphi_{00 \leftrightarrow 11}^\uparrow(W_{\frac{1}{2}}(s_{11})) = W_{\frac{1}{2}}(s_{00}).
\end{align}
But $W$ is not $T_4$-symmetric, for example consider the map $\varphi_{00 \leftrightarrow 01}: \Omega_{\oA}\times\Omega_{\oB} \to \Omega_{\oA}\times\Omega_{\oB}$, then we have
\begin{align}
\varphi_{00 \leftrightarrow 01}^\uparrow(W_{\frac{1}{2}}(s_{00})) &= W_{\frac{1}{2}}(s_{00}), \\
\varphi_{00 \leftrightarrow 01}^\uparrow(W_{\frac{1}{2}}(s_{01})) &= \dfrac{1}{2}
\begin{pmatrix}
-1 & 1 \\
1 & 1
\end{pmatrix} \notin W_{\frac{1}{2}}(S), \\
\varphi_{00 \leftrightarrow 01}^\uparrow(W_{\frac{1}{2}}(s_{10})) &= W_{\frac{1}{2}}(s_{10}), \\
\varphi_{00 \leftrightarrow 01}^\uparrow(W_{\frac{1}{2}}(s_{11})) &=
\begin{pmatrix}
0 & 1 \\
0 & 0
\end{pmatrix} \notin W_{\frac{1}{2}}(S).
\end{align}

\begin{figure}
\includegraphics[width=0.7\linewidth]{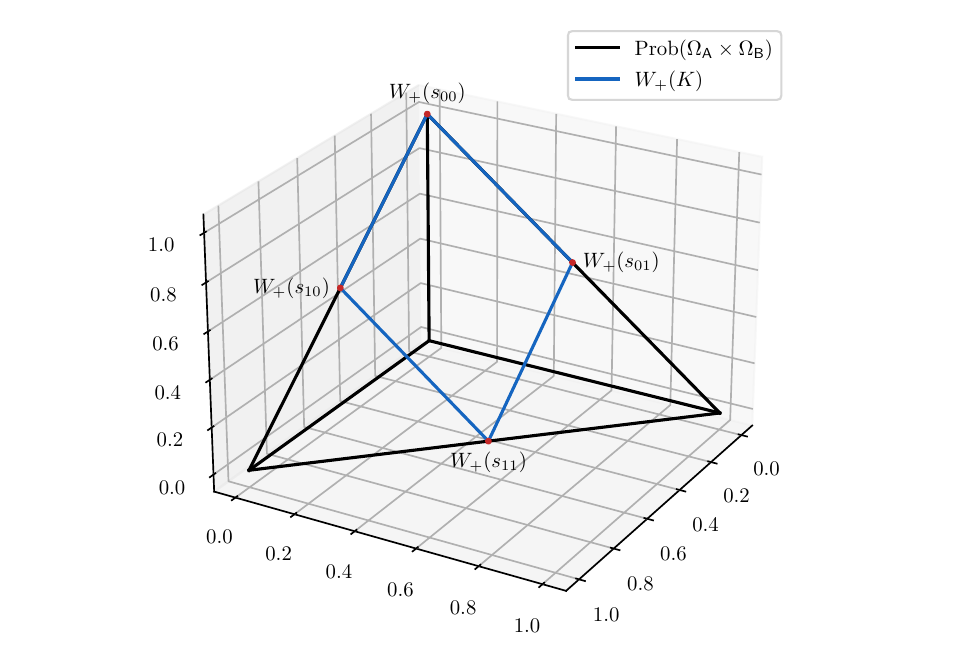}
\caption{The Wigner representation $W_+$ of $S$ with respect to $\oA$ and $\oB$.}
\label{fig:boxworld-Wplus}
\end{figure}

At last, let us present an example of a positive Wigner representation. Let $\oC$ and $\oD$ be observables with $\Omega_{\oC} = \Omega_{\oD} = \Omega = \{0,1\}$ and let $\mu_{\oC}$ and $\mu_{\oD}$ be given for $s \in S$ as $\mu_{\oC}(s) = ( \frac{1}{2} f_x(s), 1 - \frac{1}{2} f_x(s) )$, $\mu_{\oD}(s) = ( \frac{1}{2} f_y(s), 1 - \frac{1}{2} f_y(s) )$. Let $W_+: K \to \SProb(\Omega_{\oC} \times \Omega_{\oD})$ be given for $s \in S$ as
\begin{equation}
W_+(s) = 
\begin{pmatrix}
0 & \frac{1}{2} f_x(s) \\
\frac{1}{2} f_y(s) & 1 - \frac{1}{2} (f_x(s) + f_y(s))
\end{pmatrix},
\end{equation}
then we have
\begin{align}
&W_{+}(s_{00}) =
\begin{pmatrix}
0 & 0 \\
0 & 1
\end{pmatrix},
&&W_{+}(s_{01}) = \dfrac{1}{2}
\begin{pmatrix}
0 & 0 \\
1 & 1
\end{pmatrix}, \\
&W_{+}(s_{10}) = \dfrac{1}{2}
\begin{pmatrix}
0 & 1 \\
0 & 1
\end{pmatrix},
&&W_{+}(s_{11}) = \dfrac{1}{2}
\begin{pmatrix}
0 & 1 \\
1 & 0
\end{pmatrix},
\end{align}
so clearly $W_+(S) \subset \Prob(\Omega_{\oC}\times\Omega_{\oD})$. This is also easy to see from Figure \ref{fig:boxworld-Wplus}.

\section{Conclusions}
We have introduced the Wigner representations of operational theories and we have also investigated symmetries of the Wigner representations. \tao{It has been shown that under suitable assumptions, symmetry properties can uniquely fix the Wigner representation.} While we did everything in the finite-dimensional case, our results are valid for all general probabilistic theories and thus one can expect that some suitable versions of our results will also hold in the physically relevant infinite-dimensional cases.

There are several possible future continuations of our work. The first would be to construct similar treatment of angular momentum and spin. Since these would be proper physical systems, one would have to consider the role of dynamics and potentially require that time-evolution is a transported symmetry. Another avenue of research would be to consider the results highlighting negativity of Wigner functions in quantum theory as necessary for computational advantage \cite{RaussendorfBrowneDelfosseOkayBermejovega-wignerNegative} and try to obtain suitable generalizations of these results. Furthermore, it is an interesting question how the generalized Wigner representations introduced in this work relate to complex-valued quasi probability distributions \cite{Arvidsson24} like the Kirkwood-Dirac distribution.

\begin{acknowledgments}
We are thankful to Matthias Kleinmann for discussing the concepts presented in this paper.
We acknowledge support from the Deutsche Forschungsgemeinschaft (DFG, German Research Foundation, project numbers 447948357 and 440958198), the Sino-German Center for Research Promotion (Project M-0294), the ERC (Consolidator Grant 683107/TempoQ), and the German Ministry of Education and Research (Project QuKuK, BMBF Grant No. 16KIS1618K). M.P. acknowledges support from the Alexander von Humboldt Foundation.
\end{acknowledgments}

\bibliography{citations}

\end{document}